\documentclass[11pt]{article}
\usepackage[margin=1in]{geometry}
\usepackage{amsthm,amsfonts,amsmath,amssymb}
\usepackage{mathtools}
\usepackage{enumerate,graphicx}
\usepackage{xspace}
\usepackage{boxedminipage}
\usepackage{url}
\usepackage{accents}
\usepackage{multirow}
\usepackage{booktabs}
\usepackage{enumitem}
\usepackage{subcaption}
\usepackage{fullpage}
\usepackage{tikz}
\usepackage{tikz-network}
\usepackage{setspace}
\usepackage{thm-restate}
\usepackage{authblk}
\usepackage{hyperref}
\hypersetup{
    colorlinks=true,
    linkcolor=violet,
    filecolor=magenta,      
    urlcolor=cyan,
    citecolor=blue,
    pdffitwindow=true,
}\usepackage[capitalize, nameinlink]{cleveref}

\newcommand{\R}{\ensuremath{\mathbb R}}
\usepackage{algorithm}
\usepackage{algpseudocode}

\newcommand{\capacity}{\text{cap}}
\newcommand{\ind}[1]{\mathbb{I}\!\left[ #1 \right]}

\newtheorem*{theorem*}{Theorem}
\newtheorem{fact}{Fact}[section]
\newtheorem{lemma}{Lemma}[section]
\newtheorem{theorem}[lemma]{Theorem}

\newtheorem{corollary}[lemma]{Corollary}

\newtheorem{claim}[lemma]{Claim}

\renewcommand{\R}{\ensuremath{\mathbb R}}

\newcommand{\E}[1]{{\mathbb{E}}\left[#1\right]}
\newcommand{\EE}[2]{{\mathbb{E}}_{#1}\left[#2\right]}

\renewcommand{\P}[1]{{\mathbb{P}}\left[#1\right]}
\newcommand{\PP}[2]{{\mathbb{P}}_{#1}\left[#2\right]}

\theoremstyle{definition}

\newtheorem{definition}[lemma]{Definition}

\usepackage{authblk}

\usepackage[backend=bibtex, style=alphabetic, backref=true,maxbibnames=99, url=false]{biblatex} 
\begin{document}

\title{Dual Charging for Half-Integral TSP}

\author{  Nathan Klein\thanks{Boston University. Email: \textsf{nklei1@bu.edu}} \quad  Mehrshad Taziki\thanks{Email: \textsf{mtaziki18@gmail.com}} }



\maketitle

\begin{abstract}
We show that the max entropy algorithm is a randomized 1.49776 approximation for half-integral TSP, improving upon the previous known bound of 1.49993 from Karlin et al. This also improves upon the best-known approximation for half-integral TSP due to Gupta et al. 
Our improvement results from using the dual, instead of the primal, to analyze the expected cost of the matching. 
We believe this method of analysis could lead to a simpler proof that max entropy is a better-than-3/2 approximation in the general case. 

We also give a 1.4671 approximation for half integral LP solutions with no proper minimum cuts and an even number of vertices, improving upon the bound of Haddadan and Newman of 1.476. We then extend the analysis to the case when there are an odd number of vertices $n$ at the cost of an additional $O(\frac{1}{n})$ factor.
\end{abstract}

\section{Introduction}

In the metric traveling salesperson problem (TSP), we are given a weighted graph $G=(V,E)$ and aim to find the shortest closed walk that visits every vertex. Metric TSP is NP-Hard to approximate better than $\frac{123}{122}$ \cite{KLS15}. In the 1970s, Christofides and Serdyukov \cite{Chr76,Ser78} famously gave a $\frac{3}{2}$ approximation for the problem. This was not improved until recently, when Karlin, Klein and Oveis Gharan showed a $\frac{3}{2}-\epsilon$ approximation for $\epsilon=10^{-36}$ \cite{KKO21,KKO23} in 2021. Gurvits, Klein, and Leake later showed that one can set $\epsilon=10^{-34}$ \cite{GLK24}.  

These recent improvements originate in work from 2010 and 2011 by Asadpour et al. on asymmetric TSP \cite{AGMOS17} and Oveis Gharan, Saberi and Singh on graphic TSP \cite{OSS11}. These two papers introduced the so-called \textit{max entropy} algorithm for TSP. In this algorithm, one first solves the subtour elimination \cite{DFJ54} (or Held-Karp \cite{HK70}) linear program for TSP to obtain a fractional point $x \in \R_{\ge 0}^E$. Then, the max entropy distribution $\mu$ over spanning trees is computed whose marginals match $x$ (up to error $\epsilon$, which can be made exponentially small in $n$). Finally, a tree $T$ is sampled from $\mu$ and the minimum cost matching is added to the odd vertices of $T$. We know this algorithm is at worst a $\frac{3}{2}-10^{-34}$ approximation in general, and we also know there are instances where it performs no better than 1.375 \cite{JKW24}. A fascinating open question is to determine what the worst case approximation ratio of the algorithm is. 

Prior to the first improvement in the general case, Karlin, Klein and Oveis Gharan showed a randomized 1.49993 approximation for a special case of TSP. In particular, they showed that given a solution $x$ to the subtour LP with $x_e \in \{0,\frac{1}{2},1\}$ for all $e \in E$, the max entropy algorithm outputs a solution of expected cost at most $1.49993 \cdot c(x)$ \cite{KKO20}. These ``half-integral" points are of special interest due to a conjecture of Schalekamp, Williamson, and van
Zuylen \cite{SWvZ13} that the integrality gap of the subtour LP is obtained by half integral points. Notably, the lower bound of 1.375 for max entropy in \cite{JKW24} is a half integral instance, as is the classical envelope graph which demonstrates an integrality gap of at least $\frac{4}{3}$ for the subtour LP (see e.g. \cite{Williamson90}). In 2021, Gupta, Lee, Li, Mucha, Newman, and Sarkar improved the bound for half integral points to 1.49842 \cite{GLLM21} using a mix of the max entropy algorithm and a combinatorial one proposed by Haddadan and Newman \cite{HN19}.

In this work, we show that the max entropy algorithm (with no adaptations) is a 1.49776 approximation for half-integral TSP. This improves over the state of the art for half integral TSP with a significantly simpler algorithm, as well as gives a large relative improvement for the analysis of max entropy. In particular, we show:
\begin{restatable}{theorem}{maintheorem}
    \label{thm:main}
    Given an optimal solution $x$ to the subtour LP with $x_e \in \{0,\frac{1}{2},1\}$ for all $e\in E$, the max entropy algorithm produces a solution of cost at most $1.49776 \cdot c(x)$ in expectation.  
\end{restatable}
Therefore our result also bounds the integrality gap of the subtour LP in the half integral case by 1.49776 (and, should the conjecture of \cite{SWvZ13} hold, the general case as well). As discussed in more detail in \cref{sec:overview}, the primary reason for the improvement is a new dual-based analysis style. Due to the large relative improvement over the previous analysis of max entropy in the half integral case (in terms of the distance from 1.5, this is an improvement of a factor of about 30) and the fact we no longer need to analyze certain pathological cases, we believe these techniques may lead to a significant simplification of the analysis in the general case.

In \cref{sec:degreecut}, we show that a variant of the max entropy algorithm achieves the following for a special case of interest, sometimes called the half-integral degree cut case.
\begin{restatable}{theorem}{maindegreetheorem}
    \label{thm:degree-main}
    Given a half-integral solution $x$ to the subtour LP with no proper tight sets and $n$ vertices, there exists a randomized algorithm that  produces a solution of expected cost at most $(1.4671+O(\frac{1}{n})) \cdot c(x)$ when $n$ is odd and $1.4671 \cdot c(x)$ when $n$ is even.
\end{restatable}
A tight set is a set $S \subseteq V$ so that $x(\delta(S)) = 2$, and a cut $S$ is proper if $2 \le |S| \le n-2$. This improves over the ratio of 1.476 given by Haddadan and Newman \cite{HN19} for the degree cut case for $n$ even, and extends it to the case when $n$ may be odd using a method similar to \cite{GLLM21}. In this variant, inspired by \cite{HN19} and also studied in \cite{GLLM21}, before running the max entropy algorithm a random matching is sampled to perturb the marginals.

\subsection{Other Related Work}

The half-integral cycle cut case is in some sense the opposite of the degree cut case studied in \cref{thm:degree-main}, in that there are many tight sets. Cycle cut instances have the property that every tight set can be written as the union of two other tight sets. Jin, Klein, and Williamson showed a $\frac{4}{3}$ approximation algorithm and integrality gap for the half-integral cycle cut case \cite{JKW25}.

There has been exciting recent progress on two important variants of TSP, \textit{graphic} TSP, in which the input graph is unweighted, and \textit{path} TSP in which the goal is to find the shortest $s$-$t$ walk visiting every vertex. For graphic TSP, Oveis Gharan, Saberi, and Singh \cite{OSS11} first demonstrated that max entropy was a $\frac{3}{2}-\epsilon$ approximation for a small constant $\epsilon > 0$. Using different methods, M{\"o}mke and Svensson \cite{MS11} then obtained a 1.461 approximation. This was further improved by Mucha \cite{Muc12} to $\frac{13}{9}$ and then to 1.4 by Seb{\"o} and Vygen \cite{SV14}. For path TSP, Zenklusen showed a $\frac{3}{2}$ approximation \cite{Zen19} using a dynamic programming approach. Traub, Vygen, and Zenklusen then showed that given an $\alpha$ approximation for TSP there is an $\alpha+\epsilon$ approximation for Path TSP for any $\epsilon > 0$. 

\subsection{Overview}\label{sec:overview}

As discussed, in the max entropy algorithm we first solve the subtour LP to obtain a solution $x \in \R_{\ge 0}^E$. The subtour LP is as follows, where $\delta(S)$ for $S \subsetneq V$ is the set of edges with exactly one endpoint in $S$ and for $F \subseteq E$, $x(F) = \sum_{e \in F} x_e$.
\begin{equation}
\begin{aligned}
\min \hspace{2mm}& \sum\limits_{e \in E} c_e x_e \\
  \mathrm{s.t.} \hspace{2mm} &  x(\delta(S))  \ge 2 & \forall \, S\subsetneq V\\
  & x(\delta(\{v\}))  = 2 & \forall \, v\in V\\
 &   x_{e}\geq  0  &  \forall \, e \in E
\end{aligned}
\label{eq:tsplp}
\end{equation}

Then, using $x$, we find a maximum entropy distribution $\mu$ over spanning trees\footnote{Technically, over spanning trees plus an edge.} subject to the constraint that $\PP{T \sim \mu}{e \in T} = x_e$ for all $e \in E$, up to a small multiplicative error (see \cref{subsec:maxent} for more discussion on this sampling procedure). Finally, we sample a tree $T$ from $\mu$ and add a minimum cost matching $M$ on the odd vertices of the tree $T$. $T \uplus M$ is an Eulerian graph and thus contains an Eulerian walk, which is a solution to the metric TSP problem.\footnote{Using the fact that the costs form a metric, one can then shortcut the Eulerian tour to a Hamiltonian cycle of no greater cost if desired.} The main goal, then, is to \textit{bound the expected cost of the matching} over the randomness of the sampled tree $T$. To do so, we find a function that given a tree $T$ constructs a vector $y \in \R_{\ge 0}^E$ so that $c(M) \le c(y)$, and then bound the expected cost of $y$. In particular, $y$ will be a feasible point in the $O$-Join polytope (where $O$ is the set of odd vertices in $T$); see \cref{sec:prelim} for further details.  

This is the approach taken by all previous papers analyzing max entropy \cite{OSS11,KKO20,KKO21} or variants of this algorithm \cite{GLLM21}, and we do not deviate from this here. However, here we construct $y$ in a new way that streamlines the analysis and allows for a sharper guarantee.
While previous works bounded $\mathbb{E}[c(y)]$ by bounding the contribution to $y$ of \textit{each edge individually}, we instead bound the contribution of each dual variable to $\mathbb{E}[c(y)]$. By complementary slackness, the dual variables correspond to tight cuts with $x(\delta(S)) = 2$, which contain multiple edges $e$ with $x_e \in \{\frac{1}{2},1\}$. By bounding the cost of \textit{groups} of edges in terms of $\sum_{e \in \delta(S)} y_e$, the argument becomes much more flexible. Previously, analyses had to deal with pathological cases where single edges had a high $\mathbb{E}[y_e]$ and create workarounds. By looking at groups of edges, issues due to these pathological edges are  averaged out. Previously it was necessary to bound $\mathbb{E}[y_e] < \frac{x_e}{2}$ for all $e \in E$ (as this would correspond to an expected cost of less than $\frac{1}{2}c(x)$ for the matching), in our construction there may be edges with $\mathbb{E}[y_e] > \frac{x_e}{2}$. 

There is one other big advantage of this approach, which is as follows.\footnote{Understanding this advantage is a bit technical, so we recommend readers unfamiliar with work on the max entropy algorithm come back to this point after having read the body of the paper.} In the analysis of max entropy, typically one begins with $y = \frac{x_e}{2}$ for all edges. Then, depending on the parity of certain cuts in the tree, some edges have $y$ decreased and other edges have $y$ increased. In the per-edge argument, it was important to bound the expected increase of each edge carefully. However, when considering a cut, these increases and decreases often cancel out, simplifying things quite a bit. (Despite this, our analysis is not simpler than \cite{KKO20}. This is because we take care to improve the analysis in several places which require additional casework.)

One other aspect of our improvement is to incrementally improve some of the important probabilistic bounds from \cite{KKO20} using polynomial capacity (see e.g. \cite{GLK24} for usage of this tool in TSP) or more precise arguments. However the impact of this is relatively minor compared to that of the move to a dual-based analysis. Polynomial capacity, on the other hand, is the main tool used to improve the analysis of the degree cut case in \cref{sec:degreecut}.

\section{Preliminaries}\label{sec:prelim}

\subsection{Notation}

\paragraph{Sets and Cuts.} Given a set $S$, let $\delta(S)$ be the set of edges with exactly one endpoint in $S$ and $E(S)$ be the set of edges with both endpoints in $S$. Given $x \in \R^E$ and $F \subseteq E$, let $x(F) = \sum_{e \in F} x_e$. 
 A set $S\subseteq V$ is \textit{tight} if $x(\delta(S)) = 2$. A set $S$ is \textit{proper} if $2 \le |S| \le n-2$. Two sets $A,B$ \textit{cross} if $A \cap B, A \smallsetminus B, B \smallsetminus A, \overline{A \cup B}$ are all non-empty.

\paragraph{Support Graph.} We will use $x$ to construct a 4-regular and 4-edge-connected graph which we will call $G$. $G$ will have the same vertex set as our input. Then for every edge $e$ with $x_e = \frac{1}{2}$, we add the corresponding edge to the support graph. For every edge with $x_e = 1$, we will add two parallel copies of $e$ to $G$.

\paragraph{Minimum Cuts.} Tight sets are therefore minimum cuts of the support graph $G$, as they have 4 edges. It is helpful to note that since $G$ is Eulerian, every cut which is not a minimum cut has at least 6 edges. For an overview of the structure of minimum cuts, we recommend an unfamiliar reader to look at the cactus representation of Dinits, Karzanov, and Lomonozov \cite{DKL76}, or a succinct explanation of it by Frank and Fleiner \cite{FF09}. This structure is very important to our analysis.

\paragraph{Trees.} Given a tree $T$ and a set of edges $F$, we let $F_T = |T \cap F|$ denote the number of edges of $F$ contained in $T$. We will use $\delta_T(S)$ to denote the number of edges in $\delta(S)$ contained in $T$.

\paragraph{Parity.} For a sampled tree $T$, we say a set $S \subseteq V$ is \textit{even} if $\delta_T(S)$ is even and \textit{odd} otherwise.

\subsection{Polyhedral Background}

The subtour LP is given in \cref{eq:tsplp}. We will also crucially use the dual linear program in our analysis. By the parsimonious property \cite{GB93}, the subtour bound does not change after dropping the equality constraints. Thus, the dual can be seen to have the following formulation:
\begin{equation}
\begin{aligned}
\max \hspace{2mm}& 2\sum\limits_{S \subsetneq V} z_S \\
  \mathrm{s.t.} \hspace{2mm} & \sum_{S\subsetneq V \mid e \in \delta(S)} z_S \le c_e & \forall e \in E\\
 &   z_S\geq  0  &  \forall \, S \subsetneq V
\end{aligned}
\label{eq:duallp}
\end{equation}

A second important polytope is the spanning tree polytope. For any graph $G=(V,E)$, it is defined as:
\begin{equation}
\begin{aligned}
& x(E) = n-1 & \\
& x(E(S)) \leq |S|-1 &  \forall S\subseteq V\\
& x_e \geq 0 & \hspace{6ex} \forall e\in E.
\end{aligned}
\label{eq:spanningtreelp}
\end{equation}
Edmonds \cite{Edm70} proved that the extreme point solutions of this polytope are the characteristic vectors of the spanning trees of $G$. The following is a key and well-known fact relating solutions to the subtour LP and the spanning tree polytope which was exploited to give the below algorithm for half-integral TSP in \cite{KKO20}.
\begin{fact} If $x$ is a solution to the subtour LP \eqref{eq:tsplp} and $S$ 
  is a tight set, then the restriction of $x$ to edges in $E(S)$, that is, the fractional solution $\{x_e\}_{e\in E (S)}$,
is in the spanning tree polytope on  the graph $G'=(S,E(S))$. 
\end{fact}
\begin{proof}
 By counting the degree sum inside of $E(S)$, we obtain:
$$\sum_{e \in E(S)} x_e = \frac{2|S|-2}{2} = |S|-1.$$
For sets $S' \subseteq S$ which are not tight, the constraint holds similarly, giving $\sum_{e \in E(S')} x_e \le |S|-1$. 
 \end{proof}

We will also make significant use of the following characterization of the cost of the matching. 
It is well known (see \cite{EJ73}, for example) that given a metric, the following LP bounds upper bounds the cost of an integral perfect matching on a set of vertices $O$:
\begin{equation}
\begin{aligned}
\min \hspace{2mm}& \sum\limits_{S \subsetneq V} c_ey_e \\
  \mathrm{s.t.} \hspace{2mm} & y(\delta(S)) \ge 1 & \forall S \subsetneq V, |S \cap O| \text{ odd}\\
 &   y_e\geq  0  &  \forall \, e \in E
\end{aligned}
\label{eq:Ojoinlp}
\end{equation}
This is known as the $O$-Join polyhedron, and we call a feasible point $y$ in this polyhedron an \textbf{$O$-Join solution}, where $O$ is the set of vertices in the sampled tree $T$ with odd degree. In this context, note that the set of cuts $S$ which have constraints in \eqref{eq:Ojoinlp} are exactly those for which $\delta_T(S)$ is odd.

\subsection{Max Entropy Distribution}\label{subsec:maxent}

A distribution $\mu$ over spanning trees is $\lambda$-uniform if $\lambda \in \R_{\ge 0}^E$ and for every spanning tree $T$ of the graph,
$$\P{T} \propto \prod_{e \in T} \lambda_e$$
Given $\lambda$, we will let $\mu_\lambda$ denote the resulting $\lambda$-uniform distribution. 
Given a point $z$ in the spanning tree polytope, \cite{AGMOS17} show that one can efficiently find a $\lambda$-uniform tree with marginals arbitrarily close to $z$:
\begin{theorem}[\cite{AGMOS17}] Let $z$ be a point in the spanning tree polytope of a graph $G=(V,E)$. Then, for any $\epsilon > 0$, there is an algorithm running in time polynomial in the size of the graph and $\log(\frac{1}{\epsilon})$ that outputs a vector $\lambda \in \R_{\ge 0}^E$ so that the $\lambda$-uniform distribution $\mu_\lambda$ has the property:
$$\PP{T \sim \mu_{\lambda}}{e \in T} \le (1+\epsilon)z_e$$
\end{theorem}
Since $\epsilon$ can be chosen to be exponentially small in $n$, following previous work on half integral TSP, we will assume for brevity that we may set $\epsilon = 0$. This error can be handled using the stability of max entropy distributions \cite{SV19} (one can see this applied in  \cite{KKO21}). 

$\lambda$-uniform spanning tree distributions have maximum entropy over all distributions with the same marginal vector. Therefore, when one can set $\epsilon=0$, they are indeed the distributions of maximal entropy respecting the constraints. 

\subsection{Algorithm and Critical Sets}

We will study the version of the max entropy algorithm used by \cite{KKO20}. Here, we first fix an edge $e^+=(u,v)$ where $u,v$ have two parallel edges between them. Such an edge always exists in any extreme point solution (see \cite{BP91}), or as noted in \cite{KKO20} one can create such an edge by splitting a vertex in two and putting an edge of value 1 and cost 0 between its two endpoints. After fixing $e^+$, we iteratively select a minimal proper tight set $S$ that is not crossed (and $e^+ \not\in E(S)$). If we restrict $x$ to the edges inside $S$, then $x$ is in the spanning tree polytope \eqref{eq:spanningtreelp}. So, we can compute the max entropy distribution $\mu$ inside $S$, sample $T_S \sim \mu$, and contract it. When no such set exists, the graph must be a cycle $v_1,\dots,v_k$, at which point we sample a random cycle. In particular, for every $v_i,v_{i+1}$ that share two edges we sample one of them independently and uniformly at random. The sampled tree $T$ at the end is the union of all trees $T_S$ sampled during the procedure. See \cref{alg} for a complete description and \cite{KKO20} for further details.

\begin{algorithm}[h]
\caption{Algorithm for half-integral TSP \cite{KKO20}}
\begin{algorithmic}[1]
\State Let $x$ be a half-integral solution to \eqref{eq:tsplp} with an edge $e^+$ with $x_{e^+}=1$
\State Set $T = \emptyset$  
\While {there exists a proper tight set of $G$ that is not crossed (by a tight set)}
\State Let $S$ be a minimal such set such that $e^+\notin E(S)$
\State Compute the maximum entropy distribution $\mu$ of $E(S)$
\State Sample a tree from $\mu$ and add its edges to $T$ 
\State Set $G = G/S$
\EndWhile
\State Sample a uniformly random cycle from $G$ (including $e^+$) and add it to $T$
\State Compute the minimum $O$-Join $M$ on the odd nodes of $T$ and return $T \uplus M$ or its shortcut 
\end{algorithmic}
\label{alg}
\end{algorithm}

Following \cite{KKO20}, we call every tight set contracted by the algorithm a \textbf{critical set}. In addition, we call vertices critical sets. For a critical set $S$, we call $\delta(S)$ a critical cut. The collection of critical sets is a laminar family ${\cal L}$, where recall a family is laminar if for all $S,T \in {\cal L}$, $S \cap T \in \{\emptyset, S ,T\}$. 

There are two types of sets $S \in {\cal L}$. If at the moment before contraction, there are no proper minimum cuts inside $S$, then call $S$ a \textbf{degree cut}. Otherwise, call $S$ a \textbf{cycle cut}, in which case at the moment before contraction $S$ is a path with two edges between every pair of adjacent vertices. See \cref{fig:degree_cycle_cut} for examples of each type of cut in $G$, and see \cref{sub:hierarchy} for more details.

\begin{definition}[$G_S$]
    For every critical cut $S$, we let $G_S$ denote the support graph $G$ after contracting every critical cut lower than $S$ in the hierarchy as well as $V \setminus S$ into vertices. We use $G[S]$ to denote the graph $G_S \setminus w$, where $w$ is the vertex representing $V \setminus S$ after contraction.
\end{definition}

We also note here that by definition, edges which are in $G_S$ and $G_{S'}$ for critical cuts $S\not=S'$ are \textbf{independent}. This fact will be used crucially throughout the proof.

\paragraph{Remark.} We defer additional preliminaries concerning strongly Rayleigh distributions and polynomial capacity to after the proof overview.

\section{Overview of our Method}

In this section, we introduce the key tools used in our analysis and provide a high-level overview of our techniques. In \cref{sub:hierarchy}, we define the necessary notation and describe the structural properties of the hierarchy of critical cuts. Understanding this hierarchy is crucial to our analysis. In \cref{sub:Ojoin}, we describe the construction of an $O$-Join solution with a small expected cost. Furthermore, we provide an explanation of our proof and describe our use of the dual formulation (\cref{eq:duallp}) in analyzing the cost of the $O$-Join solution.

\subsection{Hierarchy of Critical Cuts}
\label{sub:hierarchy}

The algorithm constructs a natural hierarchy of min-cuts. That is, the critical cuts form a laminar family of min cuts, and therefore can be arranged in a hierarchical structure. At the bottom of this hierarchy are the singleton vertex cuts.

Let $S$ be a critical cut. Based on the structure of the contracted graph $G_S$, we classify $S$ into one of two types. Note that, we assumed $G_S$ is the support graph after contracting $V \setminus S$ and all critical cuts lower than $S$ in the hierarchy.

\begin{enumerate}
    \item \textbf{Degree cut:} If $G_S$ contains no proper min-cuts, we call $S$ a degree cut.
    \item \textbf{Cycle cut:} Otherwise, $G_S$ must form a cycle with two parallel edges between each pair of consecutive vertices. In this case, we call $S$ a cycle cut. We call the parallel edges in $G[S]$ that share their endpoints \textbf{companions}. Note that a pair of companions $e,f$ has the property that exactly one is in the tree and the event of which edge is chosen is independent of all other events. The remaining two pairs of parallel edges sharing endpoints in $G_S$ are called \textbf{cycle partners}. 
\end{enumerate}

\begin{fact}
    \label{fact:cut-types}
    Every min-cut in $G$ is either a critical cut or an interval of a cycle cut.
\end{fact}

For an edge $e$, let $S_e$ denote the minimal critical cut such that $e \in G[S_e]$. We will distinguish edges into two types depending the structure of $S_e$.

\begin{definition}[Top Edge and Bottom Edge]
    We call an edge $e$ a \textbf{top edge} if $S_e$ is a degree cut, and a \textbf{bottom edge} if $S_e$ is a cycle cut.
\end{definition}

We define a similar notation for the critical cuts.

\begin{definition}[Top Cut and Bottom Cut]
    A min cut $S$ is said to be a \textbf{top cut} if its parent in the hierarchy is a degree cut. Otherwise, $S$ is called a \textbf{bottom cut}, i.e., if its parent in the hierarchy is a cycle cut. See \cref{fig:definitions} for an example.
\end{definition}

Moreover, for an edge $e$, we define the \textbf{Last Cuts} of $e$ as the two maximal min cuts $S$ such that $e \in \delta(S)$. More precisely, let $e$ be a bottom edge where $S_e$ has child cuts $S_1, \cdots, S_k$ with two edges between $S_i$ and $S_{i+1}$ for $1 \leq i \leq k-1$. If $e$ is between $S_i$ and $S_{i+1}$, then, the last cuts of $e$ are $S_1 \cup \cdots \cup S_i$ and $S_{i+1} \cup \cdots \cup S_k$.

Note that the last cuts of a top edge are critical cuts, meanwhile, last cuts of a bottom edge are not necessarily a critical cut.

\begin{definition}[Going Higher]
    We say an edge $e \in \delta(S)$ goes higher in $S$ if the lowest critical cut $S'$ such that $S \subsetneq S'$ satisfies $e \in \delta(S')$. Additionally, when $S$ is a critical cut, by $\delta^{\uparrow}(S)$ we denote the edges in $\delta(S)$ that go higher in $S$. Similarly, $\delta^{\rightarrow}(S) = \delta(S) \setminus \delta^{\uparrow}(S)$ denotes the edges in $\delta(S)$ that don't go higher.
\end{definition}

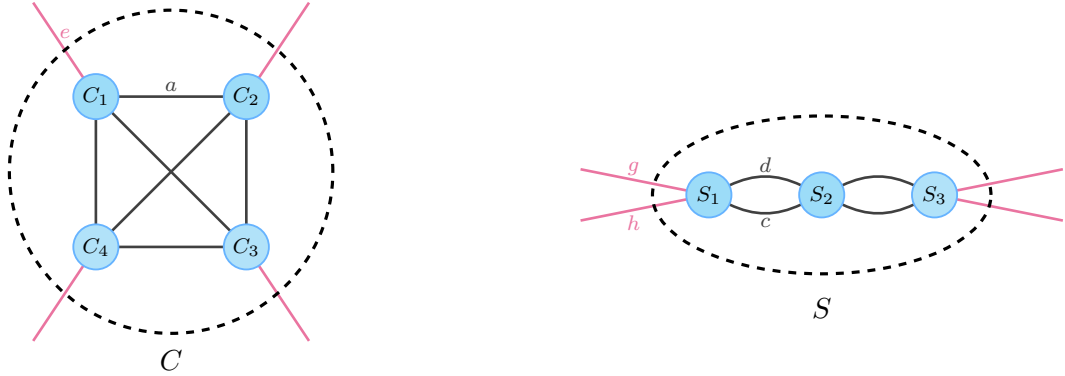
\begin{figure}[!htbp]\label{fig:definitions}
    \centering
    \begin{subfigure}[t]{0.48\textwidth}
    \centering
    \begin{tikzpicture}
    \Vertex[label = $C_1$, x = 0, y = 2, color = cyan!35, style = {draw = blue!50!cyan!60, line width = 0.7pt}]{S}
    \Vertex[label = $C_2$, style = {draw = blue!50!cyan!60, line width = 0.7pt}, x = 2, y = 2, color = cyan!35]{A}
    \Vertex[label = $C_4$, x = 0, y = 0, color = cyan!28, style = {draw = blue!50!cyan!60, line width = 0.7pt}]{B}
    \Vertex[label = $C_3$, x = 2, y = 0, color = cyan!28, style = {draw = blue!50!cyan!60, line width = 0.7pt}]{C}

    \Vertex[x=-1, y = 3.5,Pseudo]{P}
    \Vertex[x = 3, y = 3.5,Pseudo]{P1}
    \Vertex[x=-1, y = -1.5,Pseudo]{P2}
    \Vertex[x=3, y=-1.5,Pseudo]{P3}

    \Edge[label = $a$, position = above, lw = 1 pt](S)(A) 
    \Edge[position = below, lw = 1 pt](S)(B) 
    \Edge[position = below, lw = 1 pt](S)(C) 
    \Edge[position = below, lw = 1 pt](B)(A) 
    \Edge[position = below, lw = 1 pt](C)(A) 
    \Edge[position = below, lw = 1 pt](B)(C) 

    \Edge[label = $e$, color=magenta!80!purple!70, position = {above right}, lw = 1 pt](S)(P)
    
    \Edge[color=magenta!80!purple!70, position = {above left}, fontcolor = black, lw = 1 pt](A)(P1)
    
    \Edge[color=magenta!80!purple!70, position = {below right}, fontcolor = black, lw = 1 pt](B)(P2)
    
    \Edge[color=magenta!80!purple!70, position = {below left}, fontcolor = black, lw = 1 pt](C)(P3)

    \draw [black,line width=1.2pt, dashed] (1,1) ellipse (2.15 and 2.15);
    \node [draw=none] at (1,-1.5) () {$C$};
\end{tikzpicture}
    \caption{$C$ is a \textbf{degree cut} and $C_i$ are \textbf{top cuts}. The black edges are \textbf{top edges}. $e$ \textbf{goes higher} in $C_1$, while $a$ does not. The \textbf{last cuts} of $a$ are $C_1$ and $C_2$.}
\end{subfigure}
\hfill
\begin{subfigure}[t]{0.48\textwidth}
    \centering
    \raisebox{0.75cm}{
    \begin{tikzpicture}
            \Vertex[label = $S_1$, x = 0, y = 0, color = cyan!35, style = {draw = blue!50!cyan!60, line width = 0.7pt}]{A}
            \Vertex[label = $S_2$,style = {draw = blue!50!cyan!60, line width = 0.7pt}, x = 1.5, y = 0, color = cyan!35]{B}
            \Vertex[label = $S_3$,x = 3, y = 0, color = cyan!28, style = {draw = blue!50!cyan!60, line width = 0.7pt}]{C}

            \Vertex[x=-2, y=0.4,Pseudo]{P11}
            \Vertex[x=5, y=0.4,Pseudo]{P12}
            \Vertex[x=-2, y=-0.4,Pseudo]{P13}
            \Vertex[x=5, y=-0.4,Pseudo]{P14}

            \Edge[label = $c$, position = below, bend = 25, lw = 1 pt](B)(A)
            \Edge[label = $d$, position = above, bend = 25, lw = 1 pt](A)(B)
            \Edge[position = below, bend = 25, lw = 1 pt](C)(B)
            \Edge[position = above, bend = 25, lw = 1 pt](B)(C)

            \Edge[label = $g$, position = above, color=magenta!80!purple!70, lw = 1 pt](A)(P11)

            \Edge[position = above, color=magenta!80!purple!70, lw = 1 pt](C)(P12)

            \Edge[label = $h$, position = below, color=magenta!80!purple!70, lw = 1 pt](A)(P13)

            \Edge[position = below, color=magenta!80!purple!70, lw = 1 pt](C)(P14)

            \draw [black,line width=1.2pt, dashed] (1.5,0) ellipse (2.25 and 1.05);
            \node [draw=none] at (1.5,-1.5) () {$S$};
    \end{tikzpicture}}
    \caption{$S$ is a \textbf{cycle cut} and $S_i$ are \textbf{bottom cuts}. The black edges are \textbf{bottom edges}. $g$ and $h$ \textbf{go higher} in $S_1$, while $c$ and $d$ do not. $c$ and $d$ are \textbf{companions} and the \textbf{last cuts} of them are $S_1$ and $S_2 \cup S_3$. Moreover, $g, h$ are \textbf{cycle partners}.}\label{fig:degree_cycle_cut}
\end{subfigure}
\end{figure}

\subsection{Constructing the $O$-Join Solution}
\label{sub:Ojoin}

Given a tree $T$ sampled from the max-entropy distribution, we will describe a randomized process to construct a feasible solution for the $O$-Join formulation (\cref{eq:Ojoinlp}) where $O$ is the odd degree vertices of $T$. 

Before describing the construction, we will restate the definition of even at last edges from \cite{KKO20}.

\begin{definition}[Definition 4.3 in  \cite{KKO20}]
    For an edge $e$ we say $e$ is \textbf{even at last} if the two last cuts of $e$ are even. If $e$ is a bottom edge, this is equivalent to defining $e$ to be even at last when all the min cuts containing $e$ on the cycle defined by the graph consisting of $S_e$ with $V \setminus S_e$ contracted are even. If $e$ is a top edge, then it is even at last if its last cuts are even simultaneously.
\end{definition}

Let $x$ be the optimal solution of the subtour LP (\cref{eq:tsplp}). We will initialize $y = \frac{x}{2}$ so that $y$ satisfies the $O$-Join constraints. Now, when an edge $e$ is even at last, we will decrease $y_e$ by $\tau$, where $\tau$ is a parameter we will set later. Since an even at last edge can still cover lower level min-cuts in the hierarchy that are odd in $T$, we will increase the value of $y_e$ accordingly to make $y$ a feasible $O$-Join solution. 

We utilize the fact that when an edge $e$ is even at last, lower level cuts $S$ such that $e \in \delta(S)$ are (in most cases) still even with probability $\Omega(1)$. For an edge $e$, let $p_e$ denote the probability of $e$ being even at last. Unfortunately, some edges can have $p_e \approx 0$ (see \cite{KKO21}). This is an issue for arguments that bound the contribution of each edge individually. Therefore, deviating from prior work, we will instead look at the expected number of even at last edges in $\delta(S)$ for a min cut $S$. This value is denoted by $p(\delta(S)) = \sum_{e \in \delta(S)} p_e$. We show  lower bounds for this value for every min cut $S$, which in turn gives that $y(\delta(S))$ decreases  meaningfully for each min cut $S$ in the $O$-Join solution we construct. Moreover, when we increase edges to cover the violated $O$-Join constraint of a cut, we increase them according to (roughly) their $p_e$ value. This in turn means that edges that increase in the third step of our construction, should also decrease meaningfully in the second step. A more accurate and complete description of this process is provided at \cref{sec:analysis}.

Our main goal is to show that the $O$-Join solution of a tree sampled from the max entropy algorithm has cost at most $0.49776 \cdot c(x)$ in expectation. This immediately proves \cref{thm:main} as the $O$-Join polyhedron (\cref{eq:Ojoinlp}) has an integrality gap of 1. To do so, we will show the $O$-Join solution $y$ has expected cost at most $0.49776 \cdot c(x)$. Instead of bounding the contribution of each edge, we will bound the expectation on each minimum cut $S$ as follows:

\begin{restatable}{lemma}{mainlemma}
    \label{lem:main}
        There exists a randomized $O$-Join solution $y$ for the random tree $T$ sampled from the max entropy distribution such that for each min cut $S$ we have,
        \begin{align*}
            \mathbb{E}[y(\delta(S)] \leq 1 - 0.00448 = 0.99552
        \end{align*}
\end{restatable}

To analyze the cost of our solution, we utilize the dual formulation of the subtour LP (\cref{eq:duallp}). Now, we will use \cref{lem:main} to prove  \cref{thm:main}.

\maintheorem*

\begin{proof}
    
By complementary slackness we know for an edge $e$ in the support of $x$, $c_e = \sum\limits_{S:e\in \delta(S)} z_S$, therefore, the cost of the $O$-Join solution $y$ can be written as,

\begin{align*}
    c(y) = \sum_{e} c_e y_e = \sum_{e}  \sum\limits_{S:e\in \delta(S)} z_S y_e  = \sum\limits_{S:e\in \delta(S)} z_S \cdot y(\delta(S))
\end{align*}

Now, by \cref{lem:main}, 

\begin{align*}
    \mathbb{E}[c(y)] \leq \sum\limits_{S:e\in \delta(S)} 0.99552 \cdot  z_S = 0.49776 \cdot c(x) 
\end{align*}

where the final equality follows from strong duality. This gives a $1.49776$ approximation.
\end{proof}

To prove \cref{lem:main}, we analyze top cuts and bottom cuts separately. For each cut $S$, we show that either the value of $y_e$ decreases for every edge in $\delta(S)$, or, if there exists an edge for which $y_e$ does not decrease significantly, then the remaining edges in $\delta(S)$ have a larger decrease.

\section{Probabilistic and Structural Lemmas}
In this section, we will provide some key probabilistic and structural lemmas on the max entropy algorithm. These lemmas will provide strong probabilistic bounds as well as crucial observations about the structure of critical cuts that are essential in proving \cref{lem:main}.

Before doing so, in \cref{subsec:SR} and \cref{subsec:capacity} we introduce some key additional preliminaries that were omitted in \cref{sec:prelim}.

\subsection{Strongly Rayleigh Distributions}\label{subsec:SR}

Given a distribution $\mu: \mathbb{Z}_{\ge 0}^n \to \mathbb{R}_{\ge 0}$ over ground set $[n]$, its generating polynomial $g_\mu$ is defined
$$g_\mu(z) = \sum_{\kappa \in \mathbb{Z}_{\ge 0}^n} \mu(\kappa)z^\kappa$$
where $z^\kappa = \prod_{i=1}^n z_i^{\kappa_i}$. $\mu$ is \textbf{strongly Rayleigh} (SR) \cite{BBL09} if $g_\mu(z)$ is real stable. A polynomial $p(z)$ is real stable if $p(z) \not= 0$ for all $z \in \mathbb{C}^n$ with $\text{Im}(z_i) > 0$ for all $i \in [n]$. In other words, $p$ is strongly Rayleigh if it has no zeros in the upper half of the complex plane. $\lambda$-uniform spanning tree distributions are strongly Rayleigh (see e.g. \cite{BBL09,OSS11}).

\paragraph{Negative Association.} SR distributions are negatively associated \cite{BBL09, FM92}. In particular, given any increasing functions $f,g: 2^E \to \mathbb{R}$ that depend on disjoint coordinates:
$$\mathbb{E}_\mu[f] \cdot \mathbb{E}_\mu[g] \ge \mathbb{E}_\mu[f \cdot g]$$
An easy consequence is the following:
\begin{fact}[Fact 3.16 in \cite{KKO20}]
    For any $\lambda$-uniform spanning tree distribution $\mu$, for any $S \subseteq E$, any $k \ge \mathbb{R}$, and any $e \in E$ we have:
    \begin{enumerate}
        \item If $e \not\in S$ then $\PP{\mu}{e \in T \mid S_T \ge k} \le \PP{\mu}{e \in T}$ and $\PP{\mu}{e \in T \mid S_T \le k} \ge \PP{\mu}{e \in T}$.
        \item If $e \in S$, then $\PP{\mu}{e \in T \mid S_T \ge k} \ge \PP{\mu}{e \in T}$ 
        and $\PP{\mu}{e \in T \mid S_T \le k} \le \PP{\mu}{e \in T}$. 
    \end{enumerate}
\end{fact}
There is also a useful consequence of negative association when applied to a homogeneous distribution. (Recall that a polynomial is homogeneous when all terms have the same degree. A distribution is homogeneous when all outcomes have the same number of elements.)
\begin{fact}[Fact 3.17 in \cite{KKO20}]
    For any $\lambda$-uniform spanning tree distribution $\mu$, for any set of edges $S \subseteq E$ and any $e \not\in S$, we have:
    $$\EE{T \sim \mu}{S_T} \le \EE{T \sim \mu}{S_T \mid e \not\in T} \le \EE{T \sim \mu}{S_T} + x_e$$
    and similarly,
    $$\EE{T \sim \mu}{S_T}-x_e \le \EE{T \sim \mu}{S_T \mid e \in T} \le \EE{T \sim \mu}{S_T}$$
\end{fact}
When we apply one of these two facts, we will often simply say we are using negative association.

\paragraph{Closure Properties.} A second consequence of real stability is that given an SR distribution $\mu$, the following distributions are SR as well:
\begin{itemize}
    \item \textbf{Projection.} $\mu|_S$, the projection of $\mu$ to the coordinates in some $S \subseteq [n]$. 
    \item \textbf{Conditioning on a binary element to be 0 or 1.} If $z_i \in \{0,1\}$, then $\mu_{\mid z_i=0}$ and $\mu_{\mid z_i=1}$ are SR.
\end{itemize}

\paragraph{Hoeffding's Theorem.}  
For any subset $S \subseteq E$, the law of $S_T$ for $T \sim \mu$ is distributed as the sum of independent Bernoulli random variables (not necessarily all with the same success probability). This is a consequence of the fact that the coefficients of any real rooted polynomial with positive coefficients can be described by a sum of Bernoullis \cite{Pit97,BBL09}. This makes the law of $S_T$ particularly easy to analyze for any $S \subseteq E$, especially when one applies the following theorem of Hoeffding:
\begin{theorem}[{\cite[Corollary 2.1]{Hoe56}}]\label{thm:hoeffding}
Let $g:\{0,1,\dots,n\}\to \R$ and $0\leq q\leq n$ for some integer $n\geq 0$.  Let $B_1,\dots,B_n$ be $n$ independent Bernoulli random variables with success probabilities $p_1,\dots,p_n$, where $\sum_{i=1}^n p_n = q$ that minimizes (or maximizes)
$$ \E{g(B_1+\dots+B_n)}$$
over all such distributions. Then,  $p_1,\dots,p_n\in\{0,x,1\}$ for some $0<x<1$. In particular, if only $m$ of $p_i$'s are nonzero and $\ell$ of $p_i$'s are 1, then the remaining $m-\ell$ are $\frac{q-\ell}{m-\ell}$.
\end{theorem} 

A very useful corollary is the following.
\begin{lemma}[Lemma 3.23 of \cite{KKO20}] \label{lem:cut_even}
Let $S\subseteq V$ with $x(\delta(S))=2$ and $|\delta(S)| \le 4$. Then $\P{\delta_T(S)\text{ even}}\geq 13/27$.
\end{lemma}
The proof is omitted in \cite{KKO20} as it follows straightforwardly from \cref{thm:hoeffding}. One can see Lemma 3.6 in \cite{Klei24} for a proof. 

Finally, we will need the following lemma.
\begin{lemma}[Lemma 3.21 of \cite{KKO20}]
    \label{lem:three-edge}
    Let $S \subseteq E$ with $|S| = 3$. Furthermore, assume that $\mathbb{P}[|S \cap T| \geq 1] = 1.$ Then, $\mathbb{P}[|S \cap T| = 1] \geq \frac{1}{2} \quad \text{and} \quad \mathbb{P}[|S \cap T| = 2] \geq \frac{3}{8}.$
\end{lemma}

\subsection{Polynomial Capacity}\label{subsec:capacity}
The \textit{capacity} at $\alpha \in \mathbb{N}_{+}^n$ of a real stable polynomial $p(x_1,\dots,x_n)$ with positive coefficients is defined as:
$$\capacity_\alpha(p) = \inf_{x \in \R_{> 0}^n} \frac{p(x)}{x^\alpha}$$
A classical result of Gurvits \cite{Gur08} relates the capacity of a polynomial to the coefficient of $\prod_{i=1}^n x_i$ for $n$-variate homogeneous polynomials of degree $n$ as follows (where $\mathbf{1}$ is the vector consisting of $n$ 1s):
\begin{theorem}[\cite{Gur15}]\label{thm:gurvits_coefficient_bound}
    Let $p(x_1,\dots,x_n)$ be a homogeneous real stable polynomial of degree $n$ with non-negative coefficients. Then, where $C_i = \min(\deg_p(i),i)$,
    $$\frac{\partial^n}{\partial x_1 \dots \partial x_n} p\big|_{x_1=\dots=x_n=0} \ge \mathrm{cap}_{\mathbf{1}}(p) \prod_{i=2}^n \left(\frac{C_i-1}{C_i}\right)^{C_i-1}$$
\end{theorem}
Note that $\frac{\partial^n}{\partial x_1 \dots \partial x_n} p\big|_{x_1=\dots=x_n=0}$ is exactly the coefficient of $\prod_{i=1}^n x_i$. There are various similar statements in the literature, and we will use the following, first stated as Theorem 5.1 of \cite{Gur15} and restated as follows in \cite{BLP23}:
\begin{theorem}[\cite{Gur15,BLP23}]\label{thm:coeff}
    Let $p$ be a homogeneous real stable polynomial of degree $d$ with positive coefficients. Let $\alpha \in \mathbb{N}_+^n$ such that $\sum_{i=1}^n \alpha_i = d$. For $i < n$, let $d_i$ be the degree of $x_i$ in the polynomial
$$\partial_{i+1}^{\alpha_{i+1}}\dots\partial_n^{\alpha_n} \big|_{x_{i+1}=\dots=x_n=0}$$
and $d_n$ the degree of $x_n$ in $p$. Then, where $p_\alpha$ is the coefficient of the term $\prod_{i=1}^n x^{\alpha_i}$, 
$$\frac{p_\alpha}{\mathrm{cap}_\alpha(p)} \ge \prod_{i=2}^n {d_i \choose \alpha_i} \frac{\alpha_i^{\alpha_i}(d_i - \alpha_i)^{d_i-\alpha_i}}{d_i^{d_i}}$$
\end{theorem}
Furthermore, the capacity of a real stable polynomial can be bounded using its gradient. In particular, we can apply the following theorem of Gurvits and Leake \cite{GL21} (also see \cite{GLK24}) generalizing \cite{Gur06}. (We do not need the generalization here, but we state the stronger form regardless.)
\begin{theorem}\label{thm:cap}[\cite{GL21}]
    Let $p$ be a real stable polynomial in $n$ variables with non-negative coefficients, and fix any $\alpha \in \mathbb{N}_+^n$. If $p(\mathbf{1})=1$ and $\left\lVert \alpha - \nabla p(\mathbf{1})\right\lVert < 1$, then
    $$\mathrm{cap}_\alpha \ge (1-\left\lVert \alpha - \nabla p(\mathbf{1})\right\lVert)^n$$
\end{theorem}
We will use the following corollary in this work, which follows easily from the above. 
\begin{corollary}\label{cor:capacity}
    Let $p(x_1,\dots,x_n)$ be the generating polynomial of a strongly Rayleigh distribution $\mu$ over ground set $e_1,\dots,e_n$. If $\mathbb{E}[e_i] = 1$ for all $1 \le i \le n$, or equivalently $\nabla p(\mathbf{1}) = \mathbf{1}$, then if $d_i$ is the maximum degree of $x_i$,
    $$p_{\mathbf{1}} \ge \prod_{i=1}^n \frac{d_i (d_i-1)^{d_i-1}}{d_i^{d_i}}$$
\end{corollary}
\begin{proof}
Let $d$ be the maximum degree of $p$. Let $p^H(z,x_1,\dots,x_n)=z^d(x_1/z,\dots,z_n/d)$ be the homogenization of $p$. Then, $\nabla p^H(\mathbf{1}) = (d-n,1,\dots,1)$. Set $\alpha = (d-n,1,\dots,1)$. By \cref{thm:cap}, $\capacity_{(d-n,1,\dots,1)} = 1$. Applying \cref{thm:coeff} with $\alpha$, noting $\alpha_i=1$ for $i \ge 2$, we obtain:
$$p^H_\alpha \ge \prod_{i=1}^{n} \frac{d_i (d_i-1)^{d_i-1}}{d_i^{d_i}}$$
where we adjust the indices to match the degree of each variable $x_i$. But $p^H_\alpha = p_{\mathbf{1}}$, so the corollary follows.
\end{proof}

\subsection{Structure of Critical Cuts}

In this section, we recall some basic facts about the structure of critical cuts. For proofs, we refer the interested reader to Section 3 of \cite{KKO20}. These facts are used solely to ensure that our case analysis is exhaustive and covers all possible scenarios. 

\begin{fact}[Fact 3.10 \cite{KKO20}]
\label{fact:two-higher-cycle}
Suppose that $S$ is a critical set. If some (contracted) vertex $v \in S$ has two edges to $w := V \setminus S$, then $S$ is a cycle cut.
\end{fact}

\begin{fact}[Fact 3.11 \cite{KKO20}]
Suppose that $S$ and $S'$ are two distinct tight sets. Then $|\delta(S) \cap \delta(S')| \leq 2$.
\end{fact}

\begin{fact}[Fact 3.12 \cite{KKO20}]
Suppose that $S$ and $S'$ are two critical sets such that $S \subset S'$. Then if $|\delta(S) \cap \delta(S')| = 2$, then $S'$ is a cycle cut.
\end{fact}

\begin{fact}[Fact 3.13 \cite{KKO20}]
\label{fact:cons-cycle-cuts}
Suppose that $S \subset S'$ are two critical cycle cuts. Then any two edges are cycle partners on at most one of these (cycle) cuts.
\end{fact}

The following two corollaries are immediate.

\begin{corollary}
\label{cor:degree-types}
    Suppose $S$ is a degree cut. Then $S$ has at most one edge that goes higher.
\end{corollary}

\begin{corollary}
\label{cor:cycle-types}
Suppose $S$ is a cycle cut. Then $S$ has either exactly two edges or no edges that go higher.
\end{corollary}

We will also prove the following simple fact.

\begin{fact}
    \label{fact:cut-not-last}
    Let $S$ be a min cut that is not the last cut for any edge in the support graph. $S$ is always even in the tree $T$.
\end{fact}

\begin{proof}
    Let $S'$ be the parent of $S$ in the hierarchy of critical cuts. By \cref{fact:cut-types}, $S'$ must be a cycle cut with child cuts $S_1, \cdots, S_k$ with two edges between $S_i$ and $S_{i+1}$ for $1 \leq i \leq k-1$. Since $S$ is not the last cut of any edge, it must be of the form $S_i \cup \cdots \cup S_j$ for $1 < i < j < k$. 

    Since exactly one of the two companions between $S_{i-1}, S_i$ and $S_j, S_{j+1}$ are in $T$, $|\delta_T(S)| = 2$ and $S$ is even.
\end{proof}

\subsection{Bounding the Expected Number of Even at Last Edges on Min Cuts}

As previously mentioned, for an edge $e$, there cannot be any meaningful lower bound on the probability of $e$ being even at last as there can be edges with $p_e \approx 0$. In contrast, we will show that for any top cut $S$, there are strong lower bounds on $p(\delta^\rightarrow(S))$. This intuitively shows that in expectation, $y(\delta(S))$ can be decreased for every min cut. 

The following lemma can be thought of as as the analog of Lemma 5.3 of \cite{KKO20} (which showed at least one edge in each cut has $p_e \ge \frac{1}{27}$) when adapted to our framework, and uses similar proof ideas. 

\begin{lemma}
    \label{lem:4/27-1/27}
    For any top cut $S$ with no edge going higher, $p(\delta(S)) \geq \frac{4}{27}$. Moreover, if $W \subset \delta(S), |W| = 3$. Then, $p(W) \geq \frac{1}{27}$.
\end{lemma}

\begin{proof}

Let $v$ be the vertex corresponding to $S$ after contraction. For an edge $e \in \delta(S)$, with (contracted) last cuts $u$ and $v$, we have,
    \begin{align*}
        \mathbb{P}[e \text{ Even at last}] &= 1 - \mathbb{P}[\delta_T(v) \text{ is odd} \lor \delta_T(u) \text{ is odd}] \\
        &=1 - \mathbb{P}[\delta_T(u) \text{ is odd}] - \mathbb{P}[\delta_T(v) \text{ is odd}] + \mathbb{P}[\delta_T(v) \text{ is odd} \land \delta_T(u) \text{ is odd}] \\
        &\geq \frac{13}{27} - \mathbb{P}[\delta_T(v) \text{ is odd}] + \mathbb{P}[\delta_T(v) \text{ is odd} \land \delta_T(u) \text{ is odd}]
    \end{align*}
Where in the last line we used \cref{lem:cut_even}. Now we will bound $\mathbb{P}[\delta_T(v) \text{ is odd} \land \delta_T(u) \text{ is odd}]$.
\begin{align*}
    \mathbb{P}[\delta_T(v) \text{ is odd} \land \delta_T(u) \text{ is odd}] &= \mathbb{P}[\delta_T(u) \text{ is odd} \mid \delta_T(v) \text{ is odd}] \cdot \mathbb{P}[\delta_T(v) \text{ is odd}] 
    \\ &\geq \mathbb{P}\Big[\delta_T(u) \text{ is odd} \Bigm| \delta_T(v) = 1\Big] \cdot \mathbb{P}[\delta_T(v) = 1]
\end{align*}
Conditioning on $\delta_T(v) = 1$ is identical to conditioning on $\delta_T(v) \leq 1$, as $\P{\delta_T(v) \ge 1} = 1$. Therefore, by negative association, we have $\mathbb{P}\Big[e \in T \bigm| \delta_T(v) = 1\Big] \leq \frac{1}{2}$. Moreover,
$$\mathbb{P}\Big[\delta_T(u) \text{ is odd} \Bigm| \delta_T(v) = 1\Big] \geq \mathbb{P}\Big[e \in T \Bigm| \delta_T(v) = 1\Big]$$
as the measure conditioned on $\delta_T(v) = 1$ first samples a tree $T'$ from $V - v$ and then chooses an edge in $\delta_T(v)$ according to its conditional probability. Depending on the degree of $u$ in $T'$ being even or odd, we can make $\delta_T(u)$ odd by either conditioning on $e$ being in the tree or out of the tree. By negative association, $\mathbb{P}\Big[e \notin T \Bigm| \delta(v) = 1\Big] \geq \mathbb{P}\Big[e \in T \Bigm| \delta_T(v) = 1 \Big]$, therefore,
\begin{align*}
    \mathbb{P}[\delta_T(v) \text{ is odd} \land \delta_T(u) \text{ is odd}] \geq \mathbb{P}\Big[e \in T \Bigm| \delta_T(v) = 1\Big] \cdot \mathbb{P}[\delta_T(v) = 1]
\end{align*}

Which in turn gives, 
\begin{align*}
    p_e \geq \frac{13}{27} - \mathbb{P}[\delta_T(v) \text{ is odd}] + \mathbb{P}\Big[e \in T \Bigm| \delta_T(v) = 1\Big] \cdot \mathbb{P}[\delta_T(v) = 1]
\end{align*}

As $\sum\limits_{e \in \delta(v)} \mathbb{P}\Big[e \in T \Bigm| \delta_T(v) = 1\Big] = 1$, summing over $\delta(v)$ gives,
\begin{align*}
    p(\delta(v)) &\geq \frac{52}{27} - 4\cdot \mathbb{P}[\delta_T(v) \text{ is odd}] + \mathbb{P}[\delta_T(v) = 1] \\
    &= \frac{52}{27} - 3 \cdot \mathbb{P}[\delta_T(v) = 1] - 4\cdot \mathbb{P}[\delta_T(v) = 3]
\end{align*}

By \cref{thm:hoeffding}, this term attains its minimum when the Bernoullis are $\{1 , \frac{1}{3}, \frac{1}{3}, \frac{1}{3} \}$ at $\frac{4}{27}$. This shows the first claim in the lemma, that $p(\delta(S)) \ge \frac{4}{27}$. 

For the second part, consider the inequality,
\begin{align*}
    p_e \geq \frac{13}{27} - \mathbb{P}[\delta_T(S) \text{ is odd}] + \mathbb{P}\Big[e \in T \Bigm| \delta_T(v) = 1\Big] \cdot \mathbb{P}[\delta_T(v) = 1]
\end{align*}
     Since for every edge in $\delta(S)$ we have $\mathbb{P}\Big[e \in T \Bigm| \delta_T(v) = 1\Big] \leq \frac{1}{2}$, $\mathbb{E}\Big[W_T \Bigm| \delta_T(v) = 1\Big] \geq \frac{1}{2}$. By summing this inequality over all $e \in W$ we get,
\begin{align*}
    p(W) &\geq \frac{39}{27} - 3 \cdot \mathbb{P}[\delta_T(v) \text{ is odd}] + 0.5 \cdot \mathbb{P}[\delta_T(v) = 1] \\
    &= \frac{13}{9} - 2.5 \cdot \mathbb{P}[\delta_T(v) = 1] - 3 \cdot \mathbb{P}[\delta_T(v) = 3]
\end{align*}
By \cref{thm:hoeffding}, this term attains its minimum at $\frac{1}{27}$ which concludes the proof.
\end{proof}

Now, we will prove a similar result for top cuts with an edge that goes higher. But first, we will prove the following simple lemma.

\begin{lemma}
    \label{lem:13/54-edge}
    Let $e$ be a top edge with  (contracted) last cuts $u$ and $v$. Suppose $u$ has an edge going higher, and $v$ does not. Then $e$ is even at last with probability at least $\frac{13}{54}$.
\end{lemma}

\begin{proof}
    By \cref{lem:cut_even}, $v$ is even at last with probability at least $\frac{13}{27}$. Additionally, as $e$ is in a different level of the hierarchy than the edges in $\delta^\rightarrow(v)$, by letting $e$ be in or out of $T$ respectively, we can fix the parity of $u$ with probability $\frac{1}{2}$. This concludes the proof.
\end{proof}

\begin{lemma}
    \label{lem:7/32-top-cut}
    Let $S$ be a top cut with an edge $e \in \delta(S)$ that goes higher and let $W$ be any two of the edges in $\delta^\rightarrow(S)$. Then, $p(W) \geq \frac{7}{32}$.
\end{lemma}

\begin{proof}
    Denote the two edges in $W$ by $a, b$. If the other last cut of any of these two doesn't go higher, by \cref{lem:13/54-edge} $p(W) \geq \frac{13}{54}$ which satisfies the lower bound. Now, assume the other last cut of each of $a$ and $b$ has an edge that goes higher. Call these other last cuts by $S_a, S_b$ and the edges going higher from them $e_a, e_b$. Now, condition on $e \in T$. Since $e$ goes higher in $S$, it's independent from the edges that are inside $S$. Therefore, by \cref{lem:three-edge}, with probability at least $\frac{1}{2}$ exactly one of the edges in $\delta^\rightarrow(S)$ are in $T$. This makes the degree of $S$ even.

    \begin{figure}[!htpb]
    \centering
    \begin{tikzpicture}
    \Vertex[label = $S$, x = 0, y = 2, color = cyan!35, style = {draw = blue!50!cyan!60, line width = 0.7pt}]{S}
    \Vertex[style = {draw = blue!50!cyan!60, line width = 0.7pt}, x = 2, y = 2, color = cyan!35]{C}
    \Vertex[label = $S_a$, x = 0, y = 0, color = cyan!28, style = {draw = blue!50!cyan!60, line width = 0.7pt}]{A}
    \Vertex[label = $S_b$, x = 2, y = 0, color = cyan!28, style = {draw = blue!50!cyan!60, line width = 0.7pt}]{B}

    \Vertex[x=-1.5, y = 3.5,Pseudo]{P}
    \Vertex[x = 3.5, y = 3.5,Pseudo]{P3}
    \Vertex[x=-1.5, y = -1.5,Pseudo]{P1}
    \Vertex[x=3.5, y=-1.5,Pseudo]{P2}

    \Edge[label = $a$, position = left, color = green!40!teal!90](S)(A) \Edge[label = $b$, position = {below left}, color = green!40!teal!90](S)(B) \Edge[ position = {below left}](S)(C)

    \Edge[color=magenta!80!purple!70, position = {above right}, label = $e$](S)(P)
    \Edge[color=magenta!80!purple!70, position = {below right}, label = $e_a$](A)(P1)
    \Edge[color=magenta!80!purple!70, position = {below left}, label = $e_b$](B)(P2)

    \draw [black,line width=1.2pt, dashed] (1,1) ellipse (2.15 and 2.15);
    \node [draw=none] at (1,-1.5) () {$S'$};
\end{tikzpicture}
    \caption{Illustration of \cref{lem:7/32-top-cut} when $S_a$ and $S_b$ both have an edge going higher. The green edges represents the edges in $W$.}
\end{figure}
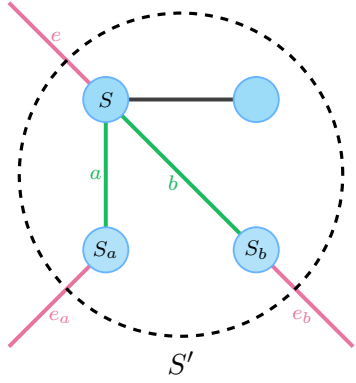

    Similarly, condition on $e \notin T$ which happens with probability $\frac{1}{2}$. By \cref{lem:three-edge}, with probability at least $\frac{3}{8}$ exactly 2 edges in $\delta^\rightarrow(S)$ are in $T$ which makes $S$ even. Thus,
    
    \begin{align*}
        \mathbb{P}[e \in T \, \land \, \delta_T(S) \text{ is even}] = \frac{1}{2} \cdot \frac{1}{2} = \frac{1}{4} 
        \quad \text{and} \quad  \mathbb{P}[e \notin T \, \land \, \delta_T(S) \text{ is even}] = \frac{1}{2} \cdot \frac{3}{8} = \frac{3}{16} 
    \end{align*}

    where we have used the fact that $e$ is independent from the edges in $\delta^\rightarrow(S)$.

    Now, in both cases, depending on the parity of $S_a, S_b$ inside $S$, we can choose $e_a$ (or $e_b$) to be inside or outside $T$ to make $a$ (or $b$) even at last. Since $\mathbb{P}[e_a \in T \mid e \in T] \leq \frac{1}{2} \leq \mathbb{P}[e_a \notin T \mid e \in T]$ and $\mathbb{P}[e_a \notin T \mid e \notin T] \leq \frac{1}{2} \leq \mathbb{P}[e_a \in T \mid e \notin T]$, we have,

    \begin{align*}
        p(W) &\geq \frac{1}{4} \Big(\mathbb{P}[e_a \in T \mid e \in T] + \mathbb{P}[e_b \in T \mid e \in T]\Big) + \frac{3}{16}\Big(\mathbb{P}[e_a \notin T \mid e \notin T] + \mathbb{P}[e_a \notin T \mid e \notin T]\Big) \\
        &\geq \frac{7}{32}
    \end{align*}

    Where the last inequality follows from  negative association.
\end{proof}

For bottom cuts we can prove stronger bounds. In fact, bottom edges inside a cycle cut are all simultaneously even at last. This symmetry enables us to individually bound $p_e$ for each edge $e$.

\begin{lemma}
    \label{lem:bottom-edge-even}
    Every bottom edge is even at last with probability at least $\frac{1}{4}$.
\end{lemma}

\begin{proof}
    Consider a bottom edge $e$, so that $S_e$ is a cycle cut with child cuts $S_1,\dots,S_k$ with two edges between each $S_i$ and $S_{i+1}$ for each $1 \le i \le k-1$. When a tree on $S_e$ is chosen, we will obtain exactly one edge among every pair of adjacent child cuts. The edges in $\delta(S)$ are comprised of two sets of edges, $A = \{a,b\} = \delta(S) \cap \delta(S_1)$ and $B = \{c,d\} = \delta(S) \cap \delta(S_k)$. So, $e$ is even at last exactly when $A_T = B_T = 1$. 

    Project $\mu$ to $\{a,b,c,d\}$. The resulting distribution has generating polynomial $p(x_a,x_b,x_c,x_d)$. Symmetrize so that $p'(x_a,x_b) = p(x_a,x_a,x_b,x_b)$. By \cref{cor:capacity}, where $d_1$ is the degree of $x_a$ and $d_2$ the degree of $x_b$,
    $$\mathbb{P}[A_T=B_T=1] = p'_{(1,1)} \ge \frac{d_1 (d_1-1)^{d_1-1}}{d_1^{d_1}} \cdot \frac{d_2 (d_2-1)^{d_2-1}}{d_2^{d_2}} \ge \frac{1}{4}$$
    as desired, since $A_T \le 2$ and $B_T \le 2$ so $d_1 \le 2$ and $d_2 \le 2$.
\end{proof}

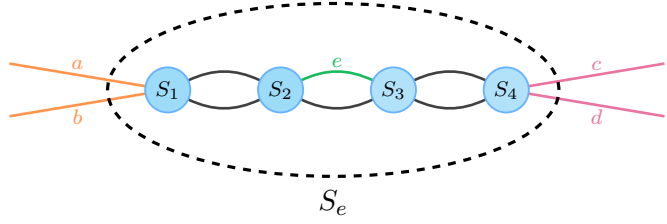
\begin{figure}[!htbp]
    \centering
    \begin{tikzpicture}[baseline=(current bounding box.center)]
            \Vertex[label = $S_1$, x = 0, y = 0, color = cyan!35, style = {draw = blue!50!cyan!60, line width = 0.7pt}]{A}
            \Vertex[label = $S_2$,style = {draw = blue!50!cyan!60, line width = 0.7pt}, x = 1.5, y = 0, color = cyan!35]{B}
            \Vertex[label = $S_3$,x = 3, y = 0, color = cyan!28, style = {draw = blue!50!cyan!60, line width = 0.7pt}]{C}
            \Vertex[label = $S_4$,x = 4.5, y = 0, color = cyan!28, style = {draw = blue!50!cyan!60, line width = 0.7pt}]{D}

            \Vertex[x=-2.4, y=0.4,Pseudo]{P11}
            \Vertex[x=6.9, y=0.4,Pseudo]{P12}
            \Vertex[x=-2.4, y=-0.4,Pseudo]{P13}
            \Vertex[x=6.9, y=-0.4,Pseudo]{P14}

            \Edge[position = below, bend = 25, lw = 1 pt](B)(A)
            \Edge[position = above, bend = 25, lw = 1 pt](A)(B)
            \Edge[position = below, bend = 25, lw = 1 pt](C)(B)
            \Edge[label = $e$, position = above, bend = 25, lw = 1 pt, color = green!40!teal!90](B)(C)
            \Edge[position = above, bend = 25, lw = 1 pt](C)(D)
            \Edge[position = below, bend = 25, lw = 1 pt](D)(C)

            \Edge[label = $a$, position = above, color=orange!80!red!70, lw = 1 pt](A)(P11)

            \Edge[label = $c$, position = above, color=magenta!80!purple!70, lw = 1 pt](D)(P12)

            \Edge[label = $b$, position = below, color=orange!80!red!70, lw = 1 pt](A)(P13)

            \Edge[label = $d$, position = below, color=magenta!80!purple!70, lw = 1 pt](D)(P14)

            \draw [black,line width=1.2pt, dashed] (2.2,0) ellipse (3 and 1.15);
            \node [draw=none] at (2.2,-1.5) () {$S_e$};
    \end{tikzpicture}
    \caption{The structure of $S_e$ at the proof of \cref{lem:bottom-edge-even}.}
\end{figure}
We remark that this lemma is tight: there are instances where this probability is exactly $\frac{1}{4}$, and the relevant edges have generating polynomial $(x_a+x_c)(\frac{1}{2}+x_b)(\frac{1}{2}+x_d)$.

\subsection{Structure of Degree Cuts and $K_5$}

In this section, we show that if a degree cut $S$ satisfies a certain structure, then $G_S$ must be the graph $K_5$. We use this fact in \cref{sec:analysis} to show that if many edges in $\delta^\rightarrow(S)$ have conditions that prevent them from decreasing reasonably, the degree cut $S$ must have a  simple structure. In particular, $G_S$ must be a $K_5$. We note that Gupta et al. \cite{GLLM21} also treated $K_5$ as a special case, and adapted the algorithm to treat these cuts differently. While we do not change the algorithm, we similarly analyze this case separately. 

\begin{lemma}
    \label{lem:K5-triangle}
    Let $S$ be a degree cut with (contracted) cuts $u,v,w \in G[S]$ each with an edge that goes higher. If, $ u,v,w$ forms a triangle, then $G_S$ is $K_5$.
\end{lemma}

\begin{proof}
    Consider the set $R = S \backslash \{ u, v, w \}$. $|\delta (R)| = 4$, so $R$ is a min-cut. However, $R \subsetneq S$, and since there are no proper min-cuts inside a degree-cut, it must be the case that $|R| = 1$. Therefore $G_S$ is a 4-regular graph with 5 vertices which concludes the proof.
\end{proof}

The following results from the fact that setting $\lambda=1$ is the $\lambda$-uniform distribution for $K_4$ when $x_e = \frac{1}{2}$ for all edges. (Note $\lambda$ is unique for a fixed $x$, so this is the only possible distribution.)
\begin{fact}
    \label{fact:K4-dist}
    The max entropy distribution on the graph $K_4$ is the uniform spanning tree distribution.
\end{fact}

We will now complement \cref{lem:K5-triangle} with the following. 

\begin{lemma}
    \label{lem:K_5-edge-even}
    Let $S$ be a degree cut of the support graph $G$ such that $G_S = K_5$. Every edge in $G[S]$ is even at last with probability at least $\frac{1}{4}$. 
\end{lemma}

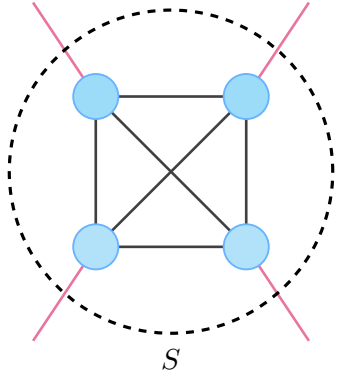
\begin{figure}[!htbp]
    \centering
    \begin{tikzpicture}
    \Vertex[x = 0, y = 2, color = cyan!35, style = {draw = blue!50!cyan!60, line width = 0.7pt}]{S}
    \Vertex[style = {draw = blue!50!cyan!60, line width = 0.7pt}, x = 2, y = 2, color = cyan!35]{A}
    \Vertex[x = 0, y = 0, color = cyan!28, style = {draw = blue!50!cyan!60, line width = 0.7pt}]{B}
    \Vertex[x = 2, y = 0, color = cyan!28, style = {draw = blue!50!cyan!60, line width = 0.7pt}]{C}

    \Vertex[x=-1, y = 3.5,Pseudo]{P}
    \Vertex[x = 3, y = 3.5,Pseudo]{P1}
    \Vertex[x=-1, y = -1.5,Pseudo]{P2}
    \Vertex[x=3, y=-1.5,Pseudo]{P3}

    \Edge[position = below, lw = 1 pt](S)(A) 
    \Edge[position = below, lw = 1 pt](S)(B) 
    \Edge[position = below, lw = 1 pt](S)(C) 
    \Edge[position = below, lw = 1 pt](B)(A) 
    \Edge[position = below, lw = 1 pt](C)(A) 
    \Edge[position = below, lw = 1 pt](B)(C) 

    \Edge[color=magenta!80!purple!70, position = {above right}, fontcolor = black, lw = 1 pt](S)(P)
    
    \Edge[color=magenta!80!purple!70, position = {above left}, fontcolor = black, lw = 1 pt](A)(P1)
    
    \Edge[color=magenta!80!purple!70, position = {below right}, fontcolor = black, lw = 1 pt](B)(P2)
    
    \Edge[color=magenta!80!purple!70, position = {below left}, fontcolor = black, lw = 1 pt](C)(P3)

    \draw [black,line width=1.2pt, dashed] (1,1) ellipse (2.15 and 2.15);
    \node [draw=none] at (1,-1.5) () {$S$};
\end{tikzpicture}
    \caption{A degree cut where $G_S$ is a $K_5$.}
    \label{fig:K4}
\end{figure}

\begin{proof}
    By \cref{fact:K4-dist}, the max entropy distribution samples a uniform spanning tree of $G[S]$. Consider an edge $uv \in G[S]$. There are $4^2 = 16$ labeled spanning trees of $K_4$ in total. For any choice of parities for $u$ and $v$, we count the number of spanning trees satisfying them,
    
    \begin{itemize}
        \item There are $2$ spanning trees where both $u$ and $v$ are even.
        \item There are $4$ spanning trees where $u$ is even and $v$ is odd.
        \item There are $4$ spanning trees where $u$ is odd and $v$ is even.
        \item There are $6$ spanning trees where both $u$ and $v$ are odd.
    \end{itemize}
    
    Since the edges going higher at $u$ and $v$ are independent of the edges in $G[S]$, we can fix the parity of $u$ and $v$ to be even by choosing an appropriate spanning tree of $G[S]$. Let $e_u, e_v$ be the edges going higher in $u$ and $v$. Then \begin{align*}
        \mathbb{P}[u, v \text{ even}] = \frac{6}{16}\mathbb{P}[\{e_u, e_v \}_T =2] + \frac{4}{16} \mathbb{P}[\{e_u, e_v \}_T =1] + \frac{2}{16} \mathbb{P}[\{e_u, e_v \}_T = 0] \geq \frac{1}{4}
\end{align*} where the inequality is by \cref{thm:hoeffding}.

This shows that degree cuts $S$, which $G_S = K_5$, are very favorable as every edge in them is even at last with large probability.
\end{proof}

\section{Analysis}
\label{sec:analysis}

\paragraph{Construction of the O-Join Solution.}
Now, we describe the construction of the O-Join solution in full detail.

For an edge $e$, recall $p_e$ is the probability of $e$ being even at last. For top edges define $\tilde{p}_e = \min(\alpha, p_e)$ and for bottom edges $\tilde{p}_e = \min(\beta, p_e)$ for constants $\alpha,\beta$ we will set momentarily. Note that $\tilde{p}(A)$ for $A \subseteq E$ denotes $\sum\limits_{e \in A} \tilde{p}_e$. For each edge $e$ with $p_e \neq 0$, let $B_e$ be an independent Bernoulli with success probability $\frac{\tilde{p}_e}{p_e}$, if $p_e = 0$, let $B_e = 0$. Moreover, for two bottom edges in the same cycle cut $e,f$, let $B_e = B_f$ with probability one. Note that since these bottom edges are always simultaneously even at last this is well-defined. 

Finally, let $S$ be a critical cut and $e \in \delta^{\rightarrow}(S)$.  By $r_S(e)$ we denote $\frac{\tilde{p}_e}{\tilde{p}(\delta^{\rightarrow}(S))}$. In other words, $r_S(e)$ represents the fraction of even at last edges in $\delta^{\rightarrow}(S)$ that $e$ is responsible for.

We will now construct $y$ in 3 different steps, done sequentially. The construction is done after sampling a tree $T$ and sampling the Bernoullis $B_e$ for every $e \in E$. 
\begin{enumerate}
    \item Let $y_e = \frac{1}{4}$ for each $e \in E$.
    \item For any even at last edge $e$ with $B_e = 1$, reduce $y_e$ by $\tau$.
    \item For each odd min cut $S$, let $\Delta(S) = \max(0, 1 - y(\delta(S)))$. For each edge $e$ and its last cuts $S, S'$, increase $e$ by $\max(r_S(e) \cdot \Delta(S), r_{S'}(e) \cdot \Delta(S'))$.
\end{enumerate}

Throughout the analysis let $\beta = \frac{1}{4}, \alpha = 0.1129032, \tau = \frac{1}{12}$. The third step of our construction along with \cref{fact:cut-not-last}, ensure that $y$ satisfies the O-Join constraint for all odd min cuts. Moreover,  since $\tau = \frac{1}{12}$, the O-Join constraint for every non min cut $S$ is also satisfied, even if all edges in $\delta(S)$ are simultaneously reduced, as such cuts have at least $6$ edges covering them.

Now, we will consider different cases for a min-cut $S$ and show that for every case $\mathbb{E}[y(\delta(S)]$ is at least reduced by $0.00448$ in expectation. For brevity, we do not consider cuts of the final cycle considered in \cref{alg}, as these cuts can easily be seen to obey the bounds described here. 

\begin{lemma}
    \label{lem:min-gain-top-cut}
    Let $S$ be a top cut with an edge that goes higher, then, $\tilde{p}(\delta(S)) \geq 2 \alpha$.
\end{lemma}
\begin{proof}
    Let $a,b,c$ be the three edges in $\delta(S)$ that don't go higher. If the probability of being even at last on none of these edges is more than $\alpha$, then,
    \begin{align*}
        \tilde{p}(\delta(S)) = p_a + p_b +p_c = \frac{(p_a + p_b) + (p_b + p_c) + (p_a + p_c)}{2} \overset{(\mathrm{i})}{\geq} \frac{3}{2} \cdot \frac{7}{32} \geq 2\alpha
    \end{align*}
    where inequality~($\mathrm{i}$) follows from \cref{lem:7/32-top-cut}. Otherwise, assume $\tilde{p}(a) = \alpha$, then, by  \cref{lem:7/32-top-cut}, $\tilde{p}(b) + \tilde{p}(c) \geq \min(\alpha, \frac{7}{32}) = \alpha$. Therefore, in all cases $\tilde{p}(\delta(S)) \geq 2\alpha$.
\end{proof}

A crucial consequence of \cref{lem:min-gain-top-cut} is that for any top cut $S$ and an edge $a \in \delta^\rightarrow(S)$, we have $r_S(a) \leq \frac{1}{2}$ as $\tilde{p}_a \leq \alpha$ by definition.

\begin{lemma}
    \label{lem:top-no-edge}    
    Let $S$ be a top cut with no edge that goes higher, then, $$\mathbb{E}[y(\delta(S)] \leq 1 - \left(\alpha + \frac{1}{27}-\frac{\beta}{4}\right) \cdot \tau$$
\end{lemma}

\begin{proof}
    Let $a,b,c,d$ be the edges in $\delta^\rightarrow(S)$ and denote the other last cuts of these edges respectively by $S_a,S_b,S_c, S_d$.
    For an edge $a$, if $S_a$ has an edge that goes higher, then by Lemma \ref{lem:13/54-edge}, $\tilde{p}_a = \alpha$. Let $e_a$ be the edge in $\delta(S_a)$ that goes higher. 
    
    If $e_a$ is a bottom edge, then 
    \begin{align*}
         \tilde{p}(e_a) = \beta \quad \text{and} \quad \mathbb{P}[S_a\text{ is odd} \mid a \text{ is reduced}] = \frac{1}{2}
    \end{align*}
     since by switching $e_a$ with its companion, we can change the parity of $S_a$ while $e_a$ remains even at last.
     
     To cover for the deficit caused by $e_a$, $a$ increases by at most $\frac{\beta}{4} \cdot \tau$ in expectation. This is because by \cref{lem:min-gain-top-cut}, $a$ is at most responsible for half of $\tilde{p}(\delta^\rightarrow(S_a))$. 
    
    Otherwise, if $e_a$ is a top edge, by  \cref{lem:three-edge} and our construction of the O-Join solution,
    \begin{align*}
        \tilde{p}(e_a) = \alpha \quad \text{and} \quad \mathbb{P}[S_a\text{ is odd} \mid a \text{ is reduced}] \leq \frac{5}{8}
    \end{align*}
     Consequently, $a$ increases by at most $\frac{5\alpha}{16} \cdot \tau \leq \frac{\beta}{4} \cdot \tau$. Therefore, for any edge $a$ such that $S_a$ has an edge that goes higher, $\mathbb{E}[y_a] \leq \frac{1}{4} - (\alpha - \frac{\beta}{4}) \cdot \tau$.

    \begin{itemize}
        \item \textbf{Case 1:} Two or more of $S_a,S_b,S_c, S_d$ have an edge that goes higher. As $\alpha - \frac{\beta}{4} > 0$, 
        \begin{align*}
        \mathbb{E}[y(\delta(S)] \leq 1 - 2 \cdot \left(\alpha - \frac{\beta}{4}\right) \cdot \tau
        \end{align*}

        \item \textbf{Case 2:} Only $S_a$ has an edge that goes higher (see \cref{fig:top-no-edge}). As above, $\mathbb{E}[y_a] \leq \frac{1}{4} - (\alpha - \frac{\beta}{4}) \cdot \tau$. For the remaining three edges, by \cref{lem:4/27-1/27}, $\tilde{p}(\{b,c,d\}) \geq \frac{1}{27}$. Therefore, 
        \begin{align*}
            \mathbb{E}[y(\delta(S)] \leq 1 - \left(\alpha + \frac{1}{27} - \frac{\beta}{4}\right) \cdot \tau
        \end{align*}

        \item \textbf{Case 3:} None of $S_a, S_b, S_c, S_d$ has an edge that goes higher. Since $\alpha \leq \frac{4}{27}$, by Lemma \ref{lem:4/27-1/27}, 
        \begin{align*}
            \mathbb{E}[y(\delta(S)] \leq 1 - \alpha \cdot \tau
        \end{align*}
    \end{itemize}

    Hence, for every case, $\mathbb{E}[y(\delta(S)] \leq 1 - \left(\alpha + \frac{1}{27}-\frac{\beta}{4}\right) \cdot \tau$
\end{proof}

\begin{figure}[!htpb]
        \centering
        \begin{tikzpicture}
            \Vertex[label = $S$, x = 0, y = 0, color = cyan!35, style = {draw = blue!50!cyan!60, line width = 0.7pt}]{S}
            \Vertex[label = $S_a$, style = {draw = blue!50!cyan!60, line width = 0.7pt}, x = -1, y = 1, color = cyan!35]{A}
            \Vertex[label = $S_b$, x = 1, y = 1, color = cyan!28, style = {draw = blue!50!cyan!60, line width = 0.7pt}]{B}
            \Vertex[label = $S_c$, x = 1, y = -1, color = cyan!28, style = {draw = blue!50!cyan!60, line width = 0.7pt}]{C}
            \Vertex[label = $S_d$, x = -1, y = -1, color = cyan!28, style = {draw = blue!50!cyan!60, line width = 0.7pt}]{D}
            
            \Vertex[x=-1.55, y = 3,Pseudo, color = cyan!28, style = {draw = blue!50!cyan!60, line width = 0.7pt}]{Pa}
            
            \Edge[label = $a$, position = {above right},lw = 1 pt](S)(A) 
            \Edge[label = $b$, position = {above left},lw = 1 pt](S)(B) 
            \Edge[label = $c$, position = {below left},lw = 1 pt](S)(C)
            \Edge[label = $d$, position = {below right},lw = 1 pt](S)(D)
            \Edge[color=magenta!80!purple!70, position = {right}, label = $e$,lw = 1 pt](A)(Pa)
            
            \draw [black,line width=1.2pt, dashed] (0,0) ellipse (1.95 and 1.95);
            \node [draw=none] at (0,-2.2) () {$S'$};
        \end{tikzpicture}
        \caption{Illustration of the worst case in \cref{lem:top-no-edge}. $S'$ is a degree cut and $e$ is a bottom edge.}
        \label{fig:top-no-edge}
\end{figure}
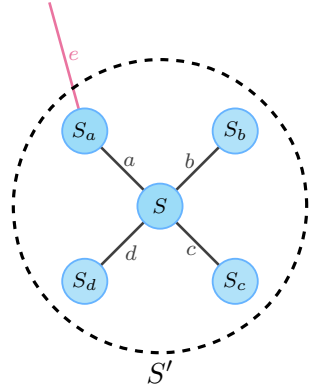

\begin{lemma}
    \label{lem:max-inc-bottom-edge}
    Let $e$ be a bottom edge inside a cycle cut $S$ with child cuts $S_1,\dots,S_k$ with two edges between each $S_i$ and $S_{i+1}$ for each $1 \le i \le k-1$. Also, let $S'$ be the parent of $S$ in the cut hierarchy. Then, we will bound the amount $e$ increases to cover the deficit on its last cuts caused by decreasing edges in $\delta(S)$ in three cases:
    \begin{enumerate}
        \item If $S'$ is a degree cut with no edge going higher, $e$ increases by at most $2\alpha \cdot \tau$.

        \item If $S'$ is a degree cut with an edge going higher, $e$ increases by at most $\left(\frac{3\alpha}{2} + \frac{\beta}{4}\right) \cdot \tau$.

        \item If $S'$ is a cycle cut, $e$ increases by at most $\left(\frac{3\beta}{4}\right) \cdot \tau$.
    \end{enumerate}
    By \cref{cor:degree-types,cor:cycle-types}, the structure of $S'$ falls into one of these three categories.
\end{lemma}

\begin{figure}[!htbp]
    \centering
    \begin{tikzpicture}[baseline=(current bounding box.center)]
            \Vertex[label = $S_1$, x = 0, y = 0, color = cyan!35, style = {draw = blue!50!cyan!60, line width = 0.7pt}]{A}
            \Vertex[label = $S_2$,style = {draw = blue!50!cyan!60, line width = 0.7pt}, x = 1.5, y = 0, color = cyan!35]{B}
            \Vertex[label = $S_3$,x = 3, y = 0, color = cyan!28, style = {draw = blue!50!cyan!60, line width = 0.7pt}]{C}
            \Vertex[label = $S_4$,x = 4.5, y = 0, color = cyan!28, style = {draw = blue!50!cyan!60, line width = 0.7pt}]{D}

            \Vertex[x = -0.02, y = 2.2,Pseudo]{P1}
            \Vertex[x=4.5, y = 2.2,Pseudo]{P2}
            \Vertex[x=4.8, y= 1.4,Pseudo]{P3}
            \Vertex[x=-1.1, y = -1.4,Pseudo]{P4}

            \Vertex[x=-2.4, y=-0.4,Pseudo]{P11}
            \Vertex[x=6.9, y=-0.4,Pseudo]{P12}

            \Edge[position = below, bend = 25, lw = 1 pt](B)(A)
            \Edge[position = above, bend = 25, lw = 1 pt](A)(B)
            \Edge[label = $f$, position = below, bend = 25, lw = 1 pt](C)(B)
            \Edge[label = $e$, position = above, bend = 25, lw = 1 pt, color = green!40!teal!90](B)(C)
            \Edge[position = above, bend = 25, lw = 1 pt](C)(D)
            \Edge[position = below, bend = 25, lw = 1 pt](D)(C)

            \Edge[color=orange!80!red!70, distance = 0.7,lw = 1 pt](A)(P11)

            \Edge[color=orange!80!red!70, lw = 1 pt](D)(P12)

            \Edge[label = $a$, color=magenta!80!purple!70, position = {left},lw = 1 pt, bend = 10](A)(P1)
            \Edge[label = $b$, color=magenta!80!purple!70, position = {right},lw = 1 pt, bend = -10](D)(P2)

            \draw [black,line width=1.2pt, dashed] (2.2,0) ellipse (3 and 1.15);
            \node [draw=none] at (2.2,-1.5) () {$S$};
            \fill[draw = blue!50!cyan!60, color = cyan!40, line width = 1.2pt, opacity =0.17] (2.24,2.48) ellipse (2.9 and 0.95);

            \draw [black,line width=1.2pt, dashed] (2.15,1.1) ellipse (3.8 and 3.4);
            \node [draw=none] at (2.2,-2.7) () {$S'$};

    \end{tikzpicture}
    \caption{An instance satisfying case 3 in Lemma \ref{lem:max-inc-bottom-edge}. Here, $L = S_1 \cup S_2$ and $R = S_3 \cup S_4$.}
\end{figure}
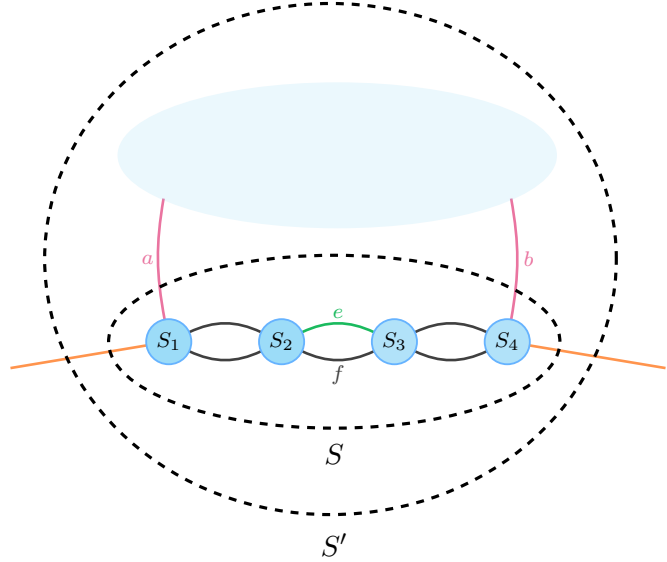

\begin{proof}

Let $f$ be the companion of $e$. $e$ and $f$ have the same set of last cuts, so any increase is divided equally on $e$ and $f$. Now:
\begin{enumerate}
    \item Since every top edge can be reduced with probability at most $\alpha$, $y_e$ can increase by at most $4 \cdot \frac{\alpha}{2} \cdot \tau= 2\alpha \cdot \tau$ to cover the deficit on its last cuts.

    \item Since the parent of $S$ is a degree cut, the edges in $\delta^\rightarrow(S)$ are top edges.  Similar to the previous case, $y_e$ increase by at most $\frac{3}{2} \cdot \alpha \cdot \tau$ because of these top edges.
    
    Let $a$ be the edge in $\delta^\uparrow(S)$ and assume $a \in \delta(S_1)$. If $a$ is a top edge, we can obtain the same bound of $2\alpha \cdot \tau$ we saw in the previous case. Otherwise, assume $a$ is a bottom edge, 
    \begin{align*}
         \tilde{p}(a) = \beta \quad \text{and} \quad \mathbb{P}[S_1\text{ is odd} \mid a \text{ is reduced}] = \frac{1}{2}
    \end{align*}
    as one can swap $a$ and its companion to change the parity of $S_1$ while retaining the even at last property for $a$. Therefore, $y_e$ increase by at most $\frac{\beta}{2} \cdot \frac{1}{2} \cdot \tau$. By the values set for $\alpha$ and $\beta$, $y_e$ can increase by at most $\left(\frac{3 \alpha}{2} \cdot \alpha + \frac{\beta}{4}\right) \cdot \tau$.

    \item The same bound proved for the charge on $e$ because of a bottom edge $a \in \delta(S)$ holds here. Let $L$ and $R$ be the last cuts of $e$, where $L$ includes $S_1$ and $R$ includes $S_k$.
    The key observation is that since $S'$ is a cycle cut, at least one pair of edges in $\delta(S)$ are companions, let $a,b$ denote those edges. By \cref{fact:cons-cycle-cuts}, one of these companions should be in $\delta(L)$ and the other in $\delta(R)$, WLOG, assume $a \in \delta(L)$ and $b \in \delta(R)$. Since companions are always simultaneously reduced, their deficit at the last cuts of $e$ is the same. Moreover, when these companions are reduced, they must be even at last. Therefore, $S$ must be even. This implies that the parity of $S_1, S_k$ is the same. It's straightforward to see the parity of $S_1$ and $S_k$ is the equal to the parity of $L$ and $R$. This implies that the events $\{ a \text{ is reduced} \, \land \,L \text{ is odd} \}$ and  $\{ b \text{ is reduced} \, \land \, R \text{ is odd} \}$ happen simultaneously.
    As increasing $e$ covers both last cuts simultaneously, we only need to account for only one of the companions.
    Each of the remaining two edges in $\delta(S)$ increases $e$ by at most $\frac{\beta}{2}$ if it's a bottom edge and $\alpha$ otherwise. Since $\frac{\beta}{2} \geq \alpha$, the former case is worse.
    This makes the total increase for $e$ at most than $\frac{3 \beta}{4} \cdot \tau$ in expectation.
\end{enumerate}
\end{proof}
  
\begin{lemma}
    \label{lem:bot-no-edge}
    Let $S$ be a bottom cut with no edge that goes higher, then, $$\mathbb{E}[y(\delta(S)] \leq 1 - (3\beta - 6\alpha) \cdot \tau$$
\end{lemma}

\begin{proof}
    By \cref{lem:bottom-edge-even}, for every bottom edge $e$, $\tilde{p}(e) = \frac{1}{4} = \beta$. Additionally, by \cref{lem:max-inc-bottom-edge} for any $e \in \delta(S)$, 
    \begin{align*}
        \mathbb{E}[y(e)] \leq \frac{1}{4} - \left(\beta - \max\Big(2\alpha, \frac{3\alpha}{2} + \frac{\beta}{4}, \frac{3 \beta}{4}\Big)\right) 
        \cdot \tau = \frac{1}{4} - \left(\frac{3\beta}{4} - \frac{3\alpha}{2}\right) \cdot \tau.
    \end{align*}
    Therefore,
    \begin{align*}
        \mathbb{E}[y(\delta(S)] \leq  1 - \left(3\beta - 6 \alpha\right)\cdot \tau
    \end{align*}
\end{proof}

\begin{figure}[!htbp]
        \centering
        \begin{tikzpicture}[scale=0.8]
            \Vertex[label = $S$, x = 2, y = 0, color = cyan!35, style = {draw = blue!50!cyan!60, line width = 0.7pt}]{S}
            \Vertex[style = {draw = blue!50!cyan!60, line width = 0.7pt}, x = 0, y = 0, color = cyan!35]{A}
            \Vertex[x = 4, y = 0, color = cyan!28, style = {draw = blue!50!cyan!60, line width = 0.7pt}]{B}
            
            \Vertex[x = -0.5, y = 2.7,Pseudo]{P1}
            \Vertex[x=4.8, y = -1.4,Pseudo]{P2}
            \Vertex[x=4.8, y= 1.4,Pseudo]{P3}
            \Vertex[x=-1.1, y = -1.4,Pseudo]{P4}
            
            \Edge[label = $a$, position = below, bend = 30, lw = 1 pt](S)(A)
            \Edge[label = $b$, position = above, bend = 30, lw = 1 pt](A)(S)
            \Edge[label = $c$, position = above, bend = 30, lw = 1 pt](S)(B)
            \Edge[label = $d$, position = below, bend = 30, lw = 1 pt](B)(S)
            \Edge[color=magenta!80!purple!70, position = {above right}, label = $e$, distance = 0.74,lw = 1 pt](A)(P1)
            \Edge[color=orange!80!red!70, position = {below right}, label = $f$, distance = 0.7,lw = 1 pt](A)(P4)
            \Edge[color=orange!80!red!70, position = {below left}, label = $g$, distance = 0.7,lw = 1 pt](B)(P2)
            \Edge[color=orange!80!red!70, position = {above left}, label = $h$, distance = 0.7,lw = 1 pt](B)(P3)
            
            \draw [black,line width=1.2pt, dashed] (2,0) ellipse (3.15 and 1.15);
            \node [draw=none] at (2,-1.5) () {$S'$};
            \draw [black,line width=1.2pt, dashed] (2.05,-0.1) ellipse (3.95 and 2.1);
        \end{tikzpicture}
        \caption{Illustration of the worst case in Lemma \ref{lem:bot-no-edge}. The edges $f,h,g$ are top edges while $e$ is a bottom edge.}
        \label{fig:bot-no-edge}
\end{figure}
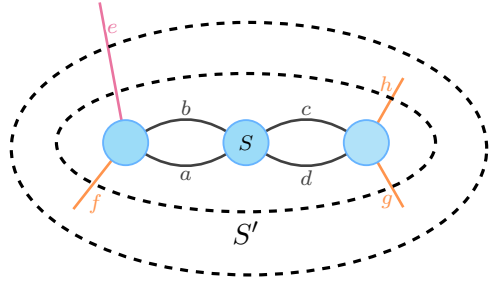

Now that we have dealt with min cuts having no edges going higher, we now turn to those with at least one such edge. The structure of these cuts is known from \cref{cor:degree-types,cor:cycle-types}.

\begin{lemma}
    \label{lem:ratio}
    Let $S$ be a top cut with an edge $e$ that goes higher. Let $a, b, c$ be the edges in $\delta^{\rightarrow}(S)$, and let $S_a$, $S_b$, and $S_c$ be the other last cuts of $a$, $b$, and $c$, respectively. If each of $S_a$, $S_b$, and $S_c$ have an edge that goes higher, then $r_{S_a}(a) \leq \frac{1}{3}$. More generally, $\tilde{p}(\delta^{\rightarrow}(S_a)) = 2\alpha + \tilde{p}_a$
\end{lemma}

\begin{proof}
    Consider the min-cut $S_a$, and let $g \in \delta^{\rightarrow}(S_a) \setminus \{a\}$. Let $R$ be the other last cuts of $g$. Suppose $R$ has an edge that goes higher. Since at most four (contracted) min-cuts can have such an edge inside a degree cut, $R$ must be either $S_b$ or $S_c$.

    In this case, the cuts $S$, $S_a$, and $R$ form a triangle. By \cref{lem:K5-triangle}, $G_S$ forms a $K_5$, and \cref{lem:K_5-edge-even} implies that $g$ (and every edge in $\delta^{\rightarrow}(S_a)$) is even at last with probability $\frac{1}{4} \geq \alpha$. Therefore, $r_{S_a}(a) = \frac{1}{3}$.

    Now, suppose that the other last cut of every edge in $\delta^{\rightarrow}(S_a) \setminus \{a\}$ does not have an edge that goes higher. In this case, by ~\cref{lem:13/54-edge}, these edges are even at last with probability at least $\frac{13}{54} \geq \alpha$. Thus concluding $r_{S_a}(a) \leq \frac{1}{3}$. Moreover, in both cases it is clear that $\tilde{p}(\delta^{\rightarrow}(S_a)) = 2\alpha + \tilde{p}_a$.
\end{proof}

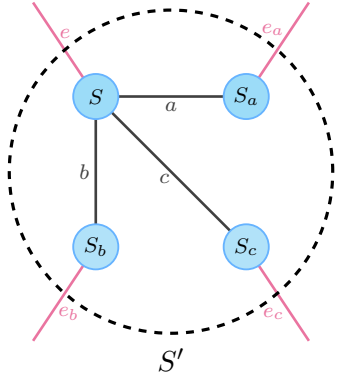
\begin{figure}[!htbp]
    \centering
    \begin{tikzpicture}
    \Vertex[label = $S$, x = 0, y = 2, color = cyan!35, style = {draw = blue!50!cyan!60, line width = 0.7pt}]{S}
    \Vertex[label = $S_a$, , style = {draw = blue!50!cyan!60, line width = 0.7pt}, x = 2, y = 2, color = cyan!35]{A}
    \Vertex[label = $S_b$, x = 0, y = 0, color = cyan!28, style = {draw = blue!50!cyan!60, line width = 0.7pt}]{B}
    \Vertex[label = $S_c$, x = 2, y = 0, color = cyan!28, style = {draw = blue!50!cyan!60, line width = 0.7pt}]{C}

    \Vertex[x=-1, y = 3.5,Pseudo]{P}
    \Vertex[x = 3, y = 3.5,Pseudo]{P1}
    \Vertex[x=-1, y = -1.5,Pseudo]{P2}
    \Vertex[x=3, y=-1.5,Pseudo]{P3}

    \Edge[label = $a$, position = below, lw = 1 pt](S)(A) 
    \Edge[label = $b$, position = left, lw = 1 pt](S)(B) 
    \Edge[label = $c$, position = {below left}, lw = 1 pt](S)(C)

    \Edge[color=magenta!80!purple!70, position = {above right}, label = $e$, lw = 1 pt](S)(P)
    
    \Edge[color=magenta!80!purple!70, position = {above left}, label = $e_a$, lw = 1 pt](A)(P1)
    
    \Edge[color=magenta!80!purple!70, position = {below right}, label = $e_b$, lw = 1 pt](B)(P2)
    
    \Edge[color=magenta!80!purple!70, position = {below left}, label = $e_{c}$, lw = 1 pt](C)(P3)

    \draw [black,line width=1.2pt, dashed] (1,1) ellipse (2.15 and 2.15);
    \node [draw=none] at (1,-1.5) () {$S'$};
\end{tikzpicture}
    \caption{A min cut satisfying the conditions of Lemma \ref{lem:ratio}}
    \label{fig:ratio}
\end{figure}

We will utilize \cref{lem:ratio} to show $\mathbb{E}[y(\delta(S))]$ decreases meaningfully when $S$ is a top cut with an edge going higher.

\begin{lemma}
    \label{lem:top-edge-higher}
    Let $S$ be a top cut with an edge that goes higher, then, $$\mathbb{E}[y(\delta(S)] \leq 1 - \left(\frac{7\alpha}{24}  +\frac{\beta}{12}\right) \cdot \tau$$
\end{lemma}

\begin{proof}
    Let $e$ be the edge in $\delta(S)$ that goes higher and let $a, b, c$ be the edges in $\delta^{\rightarrow}(S)$. Furthermore, by $S'$ we denote the parent of $S$ in the hierarchy of critical cuts and let $S_a, S_b, S_c$ be the other last cut of $a, b$ and $c$ in $S'$ respectively. We first analyze the effects of $y_e$ in $y(\delta(S))$. Let $e$ be a bottom edge, then,
    \begin{align*}
         \tilde{p}(e) = \beta \quad \text{and} \quad \mathbb{P}[S\text{ is odd} \mid e \text{ is reduced}] = \frac{1}{2}
    \end{align*}
    When $e$ decreases and $S$ is odd, we have to increase $a, b, c$ accordingly in step (3) of the construction to satisfy the O-Join constraint on $S$. However, we increase $a, b, c$ exactly as much as $e$ had decreased, therefore, $y(\delta(S))$ remains the same. Meanwhile, when $e$ is reduced and $S$ is even, which happens with probability $\frac{\beta}{2}$, $y(\delta(S))$ is reduced by $\tau$. Moreover, by \cref{lem:max-inc-bottom-edge}, $e$ itself increases by at most $\left(\frac{3\alpha}{2} + \frac{\beta}{4}\right) \cdot \tau$, which implies that $y_e$ increases $\mathbb{E}[y(\delta(S))]$ by at most $\left(\frac{3\alpha}{2} - \frac{\beta}{4}\right) \cdot \tau$.

    Now, assume $e$ is a top edge. Similar to the previous scenario, we don't need to consider the increase in $a, b, c$ because of $e$. Moreover, by \cref{lem:three-edge}, we have $\mathbb{P}[S\text{ is even} 
 \mid e \text{ is reduced}] \geq \frac{3}{8}$. So $e$ can decrease $\mathbb{E}[y(\delta(S))]$ by $\frac{3\tilde{p}_e}{8} \cdot \tau$. Furthermore, since $\frac{5\alpha}{8} < \frac{\beta}{2}$, $y_e$ can at most increase $\mathbb{E}[y(\delta(S))]$ by $\beta \cdot \frac{\tilde{p}_e}{2 \alpha} \cdot \tau$, where we have used \cref{lem:min-gain-top-cut}. It is clear that $y_e$ increases $\mathbb{E}[y(\delta(S))]$ most when $\tilde{p}_e = \alpha$, giving a total increase of $\left(\frac{\beta}{2} - \frac{3\alpha}{8}\right) \cdot \tau$. The former case is worse as $\frac{3\alpha}{2}-\frac{\beta}{4} > \frac{\beta}{2}-\frac{3\alpha}{8}$. We are now ready to analyze the possible cases.

\begin{figure}[!htbp]
    \centering
    \begin{tikzpicture}
    \Vertex[label = $S$, x = 0, y = 2, color = cyan!35, style = {draw = blue!50!cyan!60, line width = 0.7pt}]{S}
    \Vertex[label = $S_b$, , style = {draw = blue!50!cyan!60, line width = 0.7pt}, x = 2, y = 2, color = cyan!35]{B}
    \Vertex[label = $S_a$, x = 0, y = 0, color = cyan!28, style = {draw = blue!50!cyan!60, line width = 0.7pt}]{A}
    \Vertex[label = $S_c$, x = 2, y = 0, color = cyan!28, style = {draw = blue!50!cyan!60, line width = 0.7pt}]{C}

    \Vertex[x=-1, y = 3.5,Pseudo]{P}
    \Vertex[x = 3, y = 3.5,Pseudo]{P2}
    \Vertex[x=-1, y = -1.5,Pseudo]{P1}
    \Vertex[x=3, y=-1.5,Pseudo]{P3}

    \Edge[label = $a$, position = left, lw = 1 pt](S)(A) 
    \Edge[label = $b$, position = below, lw = 1 pt](S)(B) 
    \Edge[label = $c$, position = {below left}, lw = 1 pt](S)(C)

    \Edge[color=magenta!80!purple!70, position = {above right}, label = $e$,  lw = 1 pt](S)(P)
    
    \Edge[color=magenta!80!purple!70, position = {below right}, label = $e_b$, lw = 1 pt](B)(P2)
    
    \Edge[color=magenta!80!purple!70, position = {below left}, label = $e_{c}$, lw = 1 pt](C)(P3)

    \draw [black,line width=1.2pt, dashed] (1,1) ellipse (2.15 and 2.15);
    \node [draw=none] at (1,-1.5) () {$S'$};
\end{tikzpicture}
    \caption{Illustration of \textbf{Case 1} in Lemma \ref{lem:top-edge-higher}.}
    \label{fig:top-edge-higher-1}
\end{figure}
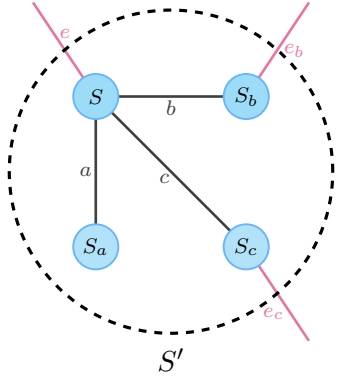

    \begin{itemize}
        \item \textbf{Case 1: At least one of $S_a, S_b, S_c$ does not have an edge going higher.} Let that edge be $a$, now, by \cref{lem:13/54-edge}, $\tilde{p}_a = \alpha$. Let $e_b, e_c$ be the (possible) edges in $S_b, S_c$ going higher (see \cref{fig:top-edge-higher-1}). By  \cref{lem:min-gain-top-cut}, $y_b$ and $y_c$ can increase by at most $\frac{\beta}{2} \cdot \frac{\tilde{p}_b}{2\alpha} \cdot \tau$ and $\frac{\beta}{2} \cdot \frac{\tilde{p}_c}{2\alpha} \cdot \tau$ in step (3) of the construction because of $e_b$ and $e_c$. Moreover, we don't need to consider the increase in $y_a$ as whenever it increases, $e$ has decreased. Thus,
        \begin{align*}
            \mathbb{E}[y(\delta(S))] \leq 1 -\left(\alpha + \tilde{p}_b + \tilde{p}_c + \frac{\beta}{2} - \Big(\frac{3\alpha}{2}+\frac{\beta}{4}\Big) - \frac{\beta}{2} \cdot \frac{\tilde{p}_b}{2\alpha} - \frac{\beta}{2} \cdot \frac{\tilde{p}_c}{2\alpha}\right) \cdot \tau
        \end{align*}
        This term attains its maximum when $\tilde{p}_b + \tilde{p}_c$ is minimized. By Lemma \ref{lem:min-gain-top-cut}, $\tilde{p}_b + \tilde{p}_c \geq \alpha$, therefore, 
        \begin{align*}
            \mathbb{E}[y(\delta(S))] \leq 1 -\left(2\alpha + \frac{\beta}{2} - \Big(\frac{3\alpha}{2}+\frac{\beta}{4}\right) - \frac{\beta}{4}\Big) \cdot \tau = 1 - \frac{\alpha}{2} \cdot \tau
        \end{align*}
        \item \textbf{Case 2: All of $S_a, S_b, S_c$ have an edge going higher and $e$ is a top edge.} Let $e_a, e_b, e_c$ denote these edges going higher. As observed previously, $e$ can increase $\mathbb{E}[y(\delta(S))]$ by at most $\left(\frac{\beta}{2} - \frac{3\alpha}{8}\right) \cdot \tau$. If at least one of $e_a, e_b, e_c$, say $e_a$, is a top edge,
        \begin{align*}
            \mathbb{E}[y(\delta(S))] \leq 1 -\left(\tilde{p}_a + \tilde{p}_b + \tilde{p}_c + \frac{3\alpha}{8} - \frac{\beta}{2} - \frac{5\alpha}{8} \cdot \frac{\tilde{p}_a}{2\alpha + \tilde{p}_a} - \frac{\beta}{2} \cdot \frac{\tilde{p}_b}{2\alpha + \tilde{p}_b} - \frac{\beta}{2} \cdot \frac{\tilde{p}_c}{2\alpha + \tilde{p}_c}\right) \cdot \tau
        \end{align*}
        where we have used \cref{lem:ratio} to bound $r_{S_a}(a), r_{S_b}(b), r_{S_c}(c)$.

        This term is maximized when $\tilde{p}_a = 0$ and $\tilde{p}_b = \tilde{p}_c = \alpha$. Remember that by \cref{lem:min-gain-top-cut}, $\tilde{p}_a + \tilde{p}_b + \tilde{p}_c \geq 2\alpha$, thus,
        \begin{align*}
            \mathbb{E}[y(\delta(S))] \leq 1 - \left(2\alpha + \frac{3\alpha}{8} - \frac{\beta}{2} - \frac{\beta}{6} - \frac{\beta}{6}\right) \cdot \tau = 1 - \left(\frac{19\alpha}{8} - \frac{5\beta}{6}\right) \cdot \tau
        \end{align*}

        Otherwise, if $e_a, e_b, e_c$ are all bottom edges, then $\tilde{p}_a, \tilde{p}_b, \tilde{p}_c = \alpha$ as we can make $S$ even with probability $\frac{13}{27}$ and then fix the parity of the other last cut by swapping bottom edges with their companion. This gives, 
        \begin{align*}
            \mathbb{E}[y(\delta(S))] \leq 1 - \left(3\alpha + \frac{3\alpha}{8} - \frac{\beta}{2} - \frac{\beta}{6} - \frac{\beta}{6} - \frac{\beta}{6}\right) \cdot \tau = 1 - \left(\frac{27\alpha}{8} - \beta \right) \cdot \tau
        \end{align*}

        \item \textbf{Case 3: All of $S_a, S_b, S_c$ have an edge going higher and $e$ is a bottom edge.}

        First, assume the companion of $e$ is not in $\delta(S')$. Now, for all edges $a,b,$ or $c$, one can set the parity of $S_a, S_b$, or $S_c$ to be even with probability $\frac{13}{27}$ and then fix the parity of $S$ by swapping $e$ and it's companion with probability $\frac{1}{2}$. Since $\alpha \leq \frac{13}{54}$, $\tilde{p}_a, \tilde{p}_b, \tilde{p}_c = \alpha$. Putting this together with our previous observations give, 
        \begin{align*}
            \mathbb{E}[y(\delta(S))] \leq 1 - \left(3\alpha + \frac{\beta}{2} - \Big(\frac{3\alpha}{2}+\frac{\beta}{4}\Big) - \frac{\beta}{6} - \frac{\beta}{6} - \frac{\beta}{6}\right) \cdot \tau = 1 - \left(\frac{3\alpha}{2} - \frac{\beta}{4}\right) \cdot \tau
        \end{align*}

        where we have also used \cref{lem:ratio} to bound $r_{S_a}(a),r_{S_b}(b),r_{S_c}(c)$.

        Now, assume $e_a$ is the companion of $e$. By \cref{fact:two-higher-cycle}, $S'$ is a bottom cut, i.e. $S''$, the parent of $S'$ in the hierarchy is a cycle cut.
        
        The aforementioned argument shows $\tilde{p}_b, \tilde{p}_c = \alpha$. However, we can not say anything about $\tilde{p}_a$ as swapping $e$ and $e_a$ changes the parity of $S_a$ and $S$ simultaneously. Note that in the worst case we can assume $\tilde{p}_a = 0$ as $y_a$ reduce $\mathbb{E}[y(\delta(S))]$ by $\left(\tilde{p}_a -\frac{\tilde{p}_a}{2\alpha + \tilde{p}_a} \cdot \frac{\beta}{2}\right) \cdot \tau$, which is minimized at $\tilde{p}_a = 0$.
        
        We consider three possible scenarios.

        \begin{enumerate}[label=(\roman*)]
          \item \textbf{$e_b, e_c$ are also companions in $S''$.}  Assume $e_b$ is reduced and $S_b$ is odd, a key observation is that $e$ is also reduced at the same time. Moreover, by swapping $e$ and $e_a$, one can make the parity of $S$ odd (see \cref{fig:top-edge-higher-i}). In other words, 
            \begin{align*}
                \mathbb{P}[e \text{ is reduced, } S \text{ is odd} \mid e_b \text{ is reduced, } S_b \text{ is odd}] = \frac{1}{2}
            \end{align*}
        The same holds for $e_c$ and $S_c$. This means that half of the increase in $b, c$ does not affect $y(\delta(S))$ as $e$ is reduced at the same time. This gives, 
        \begin{align*}
            \mathbb{E}[y(\delta(S))] \leq 1 - \left(2\alpha+\frac{\beta}{2}-\Big(\frac{3\alpha}{2}+\frac{\beta}{4}\Big)-\frac{\beta}{6}\cdot\frac{1}{2}-\frac{\beta}{6}\cdot\frac{1}{2}\right) \cdot \tau = \left(\frac{\alpha}{2}+\frac{\beta}{12}\right) \cdot \tau
        \end{align*}
        
          \item \textbf{At least one of $e_b, e_c$ is a top edge.} Say $e_c$ is a top edge, by Lemmas \ref{lem:three-edge} and \ref{lem:ratio}, $y_c$ needs to increase by at most $\frac{5 \alpha}{24} \cdot \tau$ in expectation, thus,
          \begin{align*}
              \mathbb{E}[y(\delta(S))] \leq 1 - \left(2\alpha+\frac{\beta}{2}-\Big(\frac{3\alpha}{2}+\frac{\beta}{4}\Big)-\frac{\beta}{6}-\frac{5 \alpha}{24}\right) \cdot \tau = (\frac{7\alpha}{24}+\frac{\beta}{12}) \cdot \tau
          \end{align*}
          \item \textbf{Both $e_b, e_c$ are bottom edges but not companions.} If $S''$ was a top cut, by \cref{cor:degree-types}, at least  one of $e_b, e_c$ would have been a top edge. Therefore, $S''$ must be a bottom cut in this scenario.
          The key observation here is that by Lemma \ref{lem:max-inc-bottom-edge}, $e$ increases by at most $\frac{3\beta}{4}$ here, consequently,
            \begin{align*}
                  \mathbb{E}[y(\delta(S))] \leq 1 - \left(2\alpha + \frac{\beta}{2}-\frac{3\beta}{4}-\frac{\beta}{6}-\frac{\beta}{6}\right) \cdot \tau = \left(2\alpha-\frac{7\beta}{12}\right) \cdot \tau
              \end{align*}
        \end{enumerate}
    \end{itemize}
\end{proof}

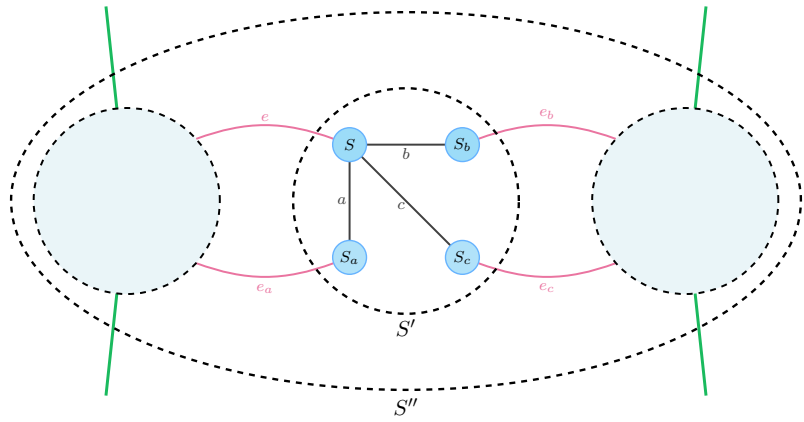
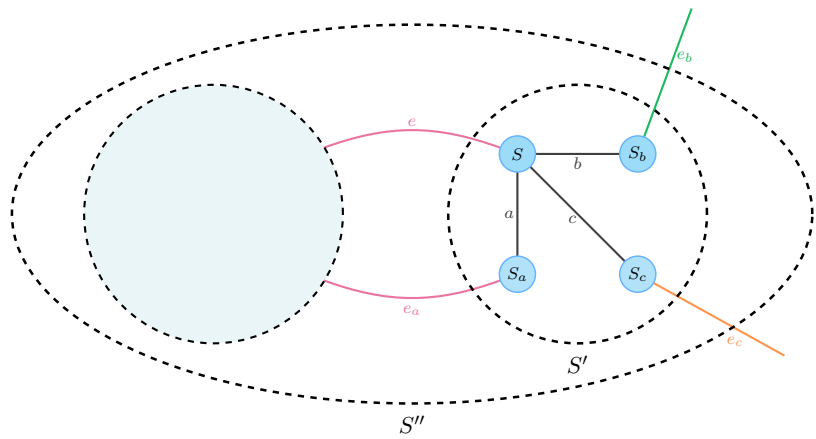
\begin{figure}[!htbp]
    \centering
    
    \begin{subfigure}[t]{\textwidth}
        \centering
        \scalebox{0.75}{
        \begin{tikzpicture}
            \Vertex[label = $S$, x = 0, y = 2, color = cyan!35, style = {draw = blue!50!cyan!60, line width = 0.7pt}]{S}
            \Vertex[label = $S_b$, , style = {draw = blue!50!cyan!60, line width = 0.7pt}, x = 2, y = 2, color = cyan!35]{B}
            \Vertex[label = $S_a$, x = 0, y = 0, color = cyan!28, style = {draw = blue!50!cyan!60, line width = 0.7pt}]{A}
            \Vertex[label = $S_c$, x = 2, y = 0, color = cyan!28, style = {draw = blue!50!cyan!60, line width = 0.7pt}]{C}
            \Vertex[x=-3, y = 2,Pseudo]{P}
            \Vertex[x=-3, y = 0,Pseudo]{P1}
            \Vertex[x = 5, y = 2,Pseudo]{P2}
            \Vertex[x = 5, y=0,Pseudo]{P3}
            \Vertex[x=-3.95, y = 1, size =3.3, style = dashed,opacity =0.25]{left}
            \Vertex[x=5.95, y = 1, size =3.3, style = dashed,opacity =0.25]{right}
            \Vertex[x=-4.35, y = 4.75,Pseudo]{O1}
            \Vertex[x=-4.35, y = -2.75,Pseudo]{O2}
            \Vertex[x = 6.35, y = 4.75,Pseudo]{O3}
            \Vertex[x = 6.35, y=-2.75,Pseudo]{O4}
            \Edge[color=green!40!teal!90](left)(O1)
            \Edge[color=green!40!teal!90](left)(O2)   
            \Edge[color=green!40!teal!90](right)(O3)
            \Edge[color=green!40!teal!90](right)(O4)
            \Edge[label = $a$, position = left, lw = 1 pt](S)(A) 
            \Edge[label = $b$, position = below, lw = 1 pt](S)(B) 
            \Edge[label = $c$, position = {below left}, lw = 1 pt](S)(C)
            \Edge[color=magenta!80!purple!70, position = {above}, label = $e$, lw = 1 pt, bend = -20](S)(P)
            \Edge[color=magenta!80!purple!70, position = {below}, label = $e_a$, lw = 1 pt, bend = 20](A)(P1)
            \Edge[color=magenta!80!purple!70, position = {above}, label = $e_b$, lw = 1 pt, distance = 0.5, bend = 20](B)(P2)
            \Edge[color=magenta!80!purple!70, position = {below}, label = $e_{c}$, lw = 1 pt, distance = 0.5, bend = -20](C)(P3)
            \draw [black,line width=1.2pt, dashed] (1,1) ellipse (2 and 2);
            \draw [black,line width=1.2pt, dashed] (1,1) ellipse (7 and 3.35);
            \node [draw=none] at (1,-1.25) () {$S'$};
            \node [draw=none] at (1,-2.65) () {$S''$};
        \end{tikzpicture}
        }
        \caption{Sub-case (i).}
        \label{fig:top-edge-higher-i}
    \end{subfigure}
    
    \begin{subfigure}[t]{\textwidth}
        \centering
        \scalebox{0.8}{
        \begin{tikzpicture}
            \Vertex[label = $S$, x = 0, y = 2, color = cyan!35, style = {draw = blue!50!cyan!60, line width = 0.7pt}]{S}
            \Vertex[label = $S_b$, , style = {draw = blue!50!cyan!60, line width = 0.7pt}, x = 2, y = 2, color = cyan!35]{B}
            \Vertex[label = $S_a$, x = 0, y = 0, color = cyan!28, style = {draw = blue!50!cyan!60, line width = 0.7pt}]{A}
            \Vertex[label = $S_c$, x = 2, y = 0, color = cyan!28, style = {draw = blue!50!cyan!60, line width = 0.7pt}]{C}
            \Vertex[x=-3.5, y = 2,Pseudo]{P}
            \Vertex[x=-3.5, y = 0,Pseudo]{P1}
            \Vertex[x = 3, y = 4.7,Pseudo]{P2}
            \Vertex[x=4.7, y=-1.5,Pseudo]{P3}
            \Vertex[x=-5.05, y = 1, size =4.3, style = dashed,opacity =0.25]{}
            \Edge[label = $a$, position = left, lw = 1 pt](S)(A) 
            \Edge[label = $b$, position = below, lw = 1 pt](S)(B) 
            \Edge[label = $c$, position = {below left}, lw = 1 pt](S)(C)
            \Edge[color=magenta!80!purple!70, position = {above}, label = $e$, lw = 1 pt, bend = -20](S)(P)
            \Edge[color=magenta!80!purple!70, position = {below}, label = $e_a$, lw = 1 pt, bend = 20](A)(P1)
            \Edge[color=green!40!teal!90, position = {below right}, label = $e_b$, lw = 1 pt, distance = 0.7](B)(P2)
            \Edge[color=orange!80!red!70, position = {below left}, label = $e_{c}$, lw = 1 pt, distance = 0.7](C)(P3)
            \draw [black,line width=1.2pt, dashed] (1,1) ellipse (2.15 and 2.15);
            \draw [black,line width=1.2pt, dashed] (-1.75,1) ellipse (6.65 and 3.15);
            \node [draw=none] at (1,-1.5) () {$S'$};
            \node [draw=none] at (-1.75,-2.5) () {$S''$};
        \end{tikzpicture}
        }
        \caption{Sub-cases (ii) and (iii).}
        \label{fig:top-edge-higher-ii}
    \end{subfigure}

    \label{fig:top-edge-higher-2}
    \caption{Illustration of sub-cases in Lemma \ref{lem:top-edge-higher}}
\end{figure}

We have proved upper bounds on $\mathbb{E}[y(\delta(S))]$ for all top cuts and all bottom cuts with no edge going higher. We next handle bottom cuts with edges going higher. By \cref{cor:cycle-types}, such cuts have exactly two edges going higher.

\begin{lemma}
    \label{lem:bot-edge-higher}
    Let $S$ be a bottom cut with two edges that go higher, then, $$\mathbb{E}[y(\delta(S)] \leq 1 - \left(\frac{2\beta}{3}-\alpha\right) \cdot \tau$$
\end{lemma}

\begin{proof}
    Let $S'$ be the parent of $S$ in the hierarchy, by $e, f$ denote the two edges in $S$ that go higher, denote the other edges in $\delta(S')$ by $g,h$. Additionally, let $C$ be the critical cut inside $S'$ which $g,h$ go higher from and let $a, b$ be the edges in $\delta^{\rightarrow}(S)$. We will now carefully analyze $\mathbb{E}[y(\delta(S))]$.

    \begin{itemize}
        \item \textbf{Case 1:} $S'$ is a bottom cut. By \cref{cor:cycle-types}, at least two of the edge in $\delta(S')$ are companions that don't go higher in $S'$, WLOG, let $e,g$ be these edges. Since they are companions, $e,g$ are even at last simultaneously. Moreover, if $e$ is even at last, then, $S'$ is even. This implies the parity of $C$ and $S$ should be the same and if $S$ is odd, $C$ should be odd too. Therefore, increasing $a, b$ simultaneously covers the deficit caused by reducing $e,g$.
        
        Another key observation is that $a,b$ in total increase exactly as much as $e$ has decreased. This implies that when $e$ (and simultaneously $g$) is even at last and $S$ is odd, $y(\delta(S))$ remains the same and we don't need to consider this scenario. By \cref{lem:max-inc-bottom-edge}, $e$ at most increases by $\left(\frac{3\alpha}{2} + \frac{\beta}{4}\right) \cdot \tau$. Meanwhile, since $\mathbb{P}[S \text{ is even} \mid e \text{ is reduced}] = \frac{1}{2}$, if can also decrease $y(\delta(S))$ by $\frac{\beta}{2} \cdot \tau$ on expectation. The same bound holds for $f$ if it is a bottom edge. Now, assume $f$ is a top edge and let $c, d$ be the edges going higher from the last cuts of $f$. By \cref{lem:min-gain-top-cut}, $f$ can at most contribute to half of the even at last edges in its last cuts. Therefore, in expectation, $f$ can at most increase by $\frac{\beta}{4}$ for each of its last cuts. Recall that one can swap $c$, or $d$, with their companion to fix the parity of the last cuts of $f$ while retaining the even at last property for $c$, or $d$. Since $\frac{\beta}{2} > \frac{3\alpha}{2} - \frac{\beta}{4}$, the latter case is worse.
        
        Finally, consider $h$, if $h$ is a bottom edge, $\tilde{p}_h = \beta$. Moreover, $\mathbb{P}[S \text{ is odd} \mid h \text{ is reduced}] = \frac{1}{2}$. Thus, $a, b$ each increase by $\frac{\beta}{4}$ in step (3) of our construction because of $h$. (Remember that $r_S(a), r_S(b) = \frac{1}{2}$.) If $h$ was a top edge, then $\tilde{p}_h = \alpha$. Since $\alpha \leq \frac{\beta}{2}$, in the worst case $h$ is a bottom edge. Using $\tilde{p}_a, \tilde{p}_b = \beta$ and all the preceding observations gives us,
        \begin{align*}
            \mathbb{E}[y(\delta(S))] \leq 1 - \left(2\beta + \frac{\beta}{2} - \Big(\frac{3\alpha}{2}+\frac{\beta}{4}\Big) - \frac{\beta}{2} - \frac{\beta}{2} \right) \cdot \tau = 1 - \left(\frac{5\beta}{4} - \frac{3\alpha}{2}\right) \cdot \tau 
        \end{align*}

        \item \textbf{Case 2:} $S'$ is a top cut. Since $S'$ is a top cut, by \cref{cor:degree-types}, at least three of $e,f,g,h$ are top edges in $\delta^{\rightarrow}(S)$. First, suppose $e,f,g$ are those edges (see \cref{subfig:h-higher}). Because of $S'$, $e$ (and similarly $f$) needs to increase by at most $\frac{\beta}{2}\cdot r_S(e) \cdot \tau = \frac{\beta}{2}\cdot \frac{\tilde{p}_e}{\tilde{p}_e + \tilde{p}_f +\tilde{p}_g} \cdot \tau$. For its other last cut, it increases by at most $\frac{\beta}{2} \cdot \frac{\tilde{p}_e}{2} \cdot \tau$ by \cref{lem:min-gain-top-cut}. Finally, $g$ can increase $a, b$ in total by at most $\tilde{p}_g$ in expectation. $\mathbb{E}[y(\delta(S))]$ is maximized when $\tilde{p}_e = \tilde{p}_f = \tilde{p}_g = \alpha$ giving
        \begin{align*}
            \mathbb{E}[y(\delta(S))] \leq 1- \left(2\beta - 2 \cdot \Big(\frac{\beta}{4}+\frac{\beta}{6}\Big)- \alpha - \frac{\beta}{2}\right)\cdot \tau = 1 - \left(\frac{2\beta}{3}-\alpha \right) \cdot \tau 
        \end{align*}
        Otherwise, WLOG, $e,g,h$ are top edges (see \cref{subfig:f-higher}). Similar to the previous case, $\mathbb{E}[y(\delta(S))]$ is maximized when $\tilde{p}_e = \tilde{p}_f = \tilde{p}_g = \alpha$. If $f$ is a bottom edge, similar to previous cases, it can decrease $\mathbb{E}[y(\delta(S))]$ by $\left(\frac{3\alpha}{2} - \frac{\beta}{4}\right) \cdot \tau$, so:
        \begin{align*}
            \mathbb{E}[y(\delta(S))] \leq 1-\left(2\beta-\Big(\frac{3\alpha}{2}-\frac{\beta}{4}\Big)-\frac{\beta}{4}-\frac{\beta}{6}-\alpha-\alpha\right)\cdot \tau = 1 - \left(\frac{11\beta}{6}-\frac{7\alpha}{2}\right) \cdot \tau 
        \end{align*}
        However, if $f$ is a top edge it can decrease $\mathbb{E}[y(\delta(S))]$ by at most $\frac{\beta}{2} \cdot \tau$. Meanwhile, the edge that goes higher from $S'$ (if such an edge exists) is a top edge, therefore, to cover for the deficit on $S'$, $e$ needs to increase by at most $\frac{5}{8} \cdot \alpha \cdot \tau$ (where we have used Lemma \ref{lem:three-edge} to show $\mathbb{P}[S\text{ is odd} \mid f \text{ is reduced}] \leq \frac{5}{8}$). This gives
        \begin{align*}
            \mathbb{E}[y(\delta(S))] \leq 1-\left(2\beta-\Big(\frac{\beta}{4}+\frac{5\alpha}{24}\Big)-\frac{\beta}{2}-\alpha-\alpha\right)\cdot \tau = 1 - \left(\frac{5\beta}{4}-\frac{53\alpha}{24}\right) \cdot \tau 
        \end{align*}
        as desired.
    \end{itemize}
        \end{proof}

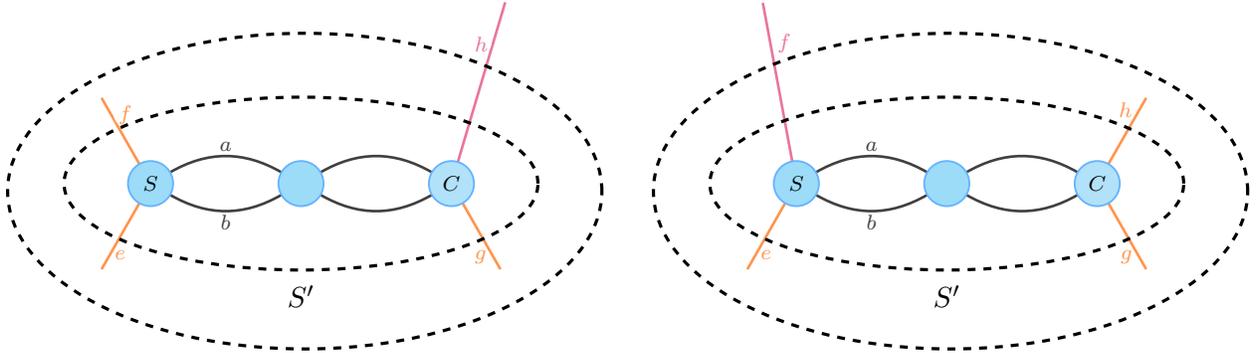
\begin{figure}[!htbp]
    \begin{subfigure}[t]{0.48\textwidth}
        \begin{tikzpicture}
            \Vertex[x = 2, y = 0, color = cyan!35, style = {draw = blue!50!cyan!60, line width = 0.7pt}]{A}
            \Vertex[label = $S$, style = {draw = blue!50!cyan!60, line width = 0.7pt}, x = 0, y = 0, color = cyan!35]{S}
            \Vertex[label = $C$, x = 4, y = 0, color = cyan!28, style = {draw = blue!50!cyan!60, line width = 0.7pt}]{B}
            
            \Vertex[x = -0.8, y = 1.4,Pseudo]{P1}
            \Vertex[x=4.8, y = -1.4,Pseudo]{P2}
            \Vertex[x=4.8, y= 2.7,Pseudo]{P3}
            \Vertex[x=-0.8, y = -1.4,Pseudo]{P4}
            
            \Edge[label = $a$, position = above, bend = 30, lw = 1 pt](S)(A)
            \Edge[label = $b$, position = below, bend = 30, lw = 1 pt](A)(S)
            \Edge[position = above, bend = 30, lw = 1 pt](B)(A)
            \Edge[ position = below, bend = 30, lw = 1 pt](A)(B)
            \Edge[color=orange!80!red!70, position = {right}, label = $f$, distance = 0.76,lw = 1 pt](S)(P1)
            \Edge[color=orange!80!red!70, position = {below right}, label = $e$, distance = 0.7,lw = 1 pt](S)(P4)
            \Edge[color=orange!80!red!70, position = {below left}, label = $g$, distance = 0.7,lw = 1 pt](B)(P2)
            \Edge[color=magenta!80!purple!70, position = {above left}, label = $h$, distance = 0.7,lw = 1 pt](B)(P3)
            
            \draw [black,line width=1.2pt, dashed] (2,0) ellipse (3.15 and 1.15);
            \node [draw=none] at (2,-1.5) () {$S'$};
            \draw [black,line width=1.2pt, dashed] (2.05,-0.1) ellipse (3.95 and 2.1);
        \end{tikzpicture}
        \caption{The scenario where $h$ is the edge in $\delta(S')$ that goes higher. As $\tilde{p}_g$ goes up, $y_a, y_b$ increase more to cover the deficit caused on $C$ by reducing $g$. Meanwhile, this makes $\tilde{p}(\delta^\rightarrow(S'))$ larger which in turn decreases the amount $y_f$ and $y_e$ go up.}
        \label{subfig:h-higher}
    \end{subfigure}
    \hfill
    \begin{subfigure}[t]{0.48\textwidth}
        \begin{tikzpicture}
            \Vertex[x = 2, y = 0, color = cyan!35, style = {draw = blue!50!cyan!60, line width = 0.7pt}]{A}
            \Vertex[label = $S$, style = {draw = blue!50!cyan!60, line width = 0.7pt}, x = 0, y = 0, color = cyan!35]{S}
            \Vertex[label = $C$, x = 4, y = 0, color = cyan!28, style = {draw = blue!50!cyan!60, line width = 0.7pt}]{B}
            
            \Vertex[x = -0.5, y = 2.7,Pseudo]{P1}
            \Vertex[x=4.8, y = -1.4,Pseudo]{P2}
            \Vertex[x=4.8, y= 1.4,Pseudo]{P3}
            \Vertex[x=-0.8, y = -1.4,Pseudo]{P4}
            
            \Edge[label = $a$, position = above, bend = 30, lw = 1 pt](S)(A)
            \Edge[label = $b$, position = below, bend = 30, lw = 1 pt](A)(S)
            \Edge[position = above, bend = 30, lw = 1 pt](B)(A)
            \Edge[ position = below, bend = 30, lw = 1 pt](A)(B)
            \Edge[color=magenta!80!purple!70, position = {right}, label = $f$, distance = 0.76,lw = 1 pt](S)(P1)
            \Edge[color=orange!80!red!70, position = {below right}, label = $e$, distance = 0.7,lw = 1 pt](S)(P4)
            \Edge[color=orange!80!red!70, position = {below left}, label = $g$, distance = 0.7,lw = 1 pt](B)(P2)
            \Edge[color=orange!80!red!70, position = {above left}, label = $h$, distance = 0.7,lw = 1 pt](B)(P3)
            
            \draw [black,line width=1.2pt, dashed] (2,0) ellipse (3.15 and 1.15);
            \node [draw=none] at (2,-1.5) () {$S'$};
            \draw [black,line width=1.2pt, dashed] (2.05,-0.1) ellipse (3.95 and 2.1);
        \end{tikzpicture}
        \caption{The scenario where $f$ is the edge in $\delta(S')$ that goes higher. When $f$ is a top edge, we lose more on $f$ but less on $e$. On the contrary, when $f$ is a bottom edge, $y_f$ increases less but $y_e$ increase more.}
        \label{subfig:f-higher}
    \end{subfigure}
    \caption{Illustration of two scenarios in \textbf{Case 2} of the proof of Lemma \ref{lem:bot-edge-higher}.}
\end{figure}

Now, we are ready to prove the main lemma of the paper.

\mainlemma*

\begin{proof}
    By \cref{cor:degree-types,cor:cycle-types}, every min-cut falls into one of the \cref{lem:top-no-edge,lem:top-edge-higher,lem:bot-no-edge,lem:bot-edge-higher}. The rest of the proof follows from the values set for $\alpha, \beta, \tau$.
\end{proof}

\section{Improving the Half-Integral Degree Cut Case}\label{sec:degreecut}
Let $x$ be the optimal solution to the subtour LP with $x_e \in \{ 0, \frac{1}{2}, 1\}$ for all $e \in E$ and let $G$ be the support graph associated with $x$. We say $x$ is a \textbf{degree-cut instance} of TSP when $G$ has no proper min-cut: in other words, if $x(\delta(S)) > 2$ for all sets $S$ with $2 \le |S| \le n-2$. Previously, Haddadan and Newman \cite{HN19} gave a $1.5-\frac{1}{42} \approx 1.476$ approximation for half-integral degree cut instances with an even number of vertices. In this section, we improve this to an $1.4671$ approximation and show a $1.4671 + O(\frac{1}{n})$ approximation for half-integral degree cut instances with an odd number of vertices. We only prove the result for the case when $n$ is odd as the other case follows similarly. We note here that our construction of the perturbed point in the odd case is similar to \cite{GLLM21} and does not lead to any significant differences. We use this construction because we believe it may generalize slightly more naturally to the non-half-integral case.

\maindegreetheorem*

We analyze a version of max entropy similar to that seen in \cite{GLLM21}: we first sample a random matching to update $x$ to a new point $x'$ in the spanning tree polytope (as is done in \cite{HN19}) and then apply max entropy on the marginals $x'$. This section serves to show that this ``perturbed" max entropy algorithm is competitive with the combinatorial algorithm of \cite{HN19}. This is some evidence that the perturbed max entropy algorithm may be worth studying further.

The matching polytope is defines as follows:

\begin{equation}
\begin{aligned}
\hspace{2mm} & x(\delta(v)) \le 1 & \forall v \in V \\
\hspace{2mm} & x\left(E(S)\right) \le \frac{|S|-1}{2} & \forall U\subseteq V, |U| \text{ odd}\\
 & x_e\geq  0  &  \forall \, e \in E
\end{aligned}
\label{eq:Ojoinlp}
\end{equation}

For odd $n$, let $\nu = \frac{n-1}{2n} \cdot x$,\footnote{In the case where there are an even number of vertices, we instead let $\nu = \frac{1}{2}x$.} and note that in the degree cut case $x_e = \frac{1}{2}$ for all edges not equal to 0, as an edge with $x_e = 1$ would induce a proper tight set.

\begin{claim}
    $\nu$ is in the matching polytope of $G$. Therefore, it can be expressed as a convex combination of maximum matchings in $G$.
\end{claim}

\begin{proof}
For any proper set $S \subseteq V$, $x(E(S)) \le \frac{|S|-1}{2}$, by the subtour elimination constraints. The fact that all matchings in the convex combination are maximal follows from the fact that $\sum_e \nu_e = \frac{n-1}{2}$, and no matching can have size $\frac{n}{2}$ as $n$ is odd.  
\end{proof}

 Let $M$ be a random matching sampled from $\nu$. Then, by the above claim:

\begin{fact}
    Every vertex $v \in G$ is saturated in $M$ with probability $1 - \frac{1}{n}$.
\end{fact}

We construct a new point $z$ using the randomly sampled matching $M$. $z$ will be the marginals for sampling a spanning tree plus an edge. For each $e = uv \in G$ define:
\begin{align*}
    z_e = \frac{1}{2} + \frac{2}{3} \ind{e \in M} - \frac{1}{12} \left(\ind{u \text{ is saturated in } M} + \ind{v \text{ is saturated in } M}\right)
\end{align*}

In other words, let $r$ be the single vertex not saturated in $M$. Set $z_e = 1$ for all $e \in M$; for $e \in \delta(r)$, set $z_e = \frac{5}{12}$ and for all the other edges $z_e = \frac{1}{3}$. Let $e^+$ be an arbitrary edge in $M$. For every $e \in G \setminus e^+$, let $z'_e = z_e$ and let $z'_{e^+} = 0$.

\begin{claim}
    \label{lem:degree-z-expected}
    $z'$ is in the spanning tree polytope \eqref{eq:spanningtreelp} of $M$. Furthermore, $\mathbb{E}[z_e] = \frac{1}{2}$ for all $e \in E$.
\end{claim}

\begin{proof}
    We begin by computing $\mathbb{E}[z_e] = \frac{1}{2}$ for all edges $e \in G$. With probability $\frac{2}{n}$, one end point of $e$ is $r$, and hence, $z_e = \frac{5}{12}$. Otherwise, with probability $\frac{n-1}{4n}$, $e \in M$ and $z_e = 1$. Otherwise, $z_e = \frac{1}{3}$. This demonstrates: 
    \begin{align*}
        \mathbb{E}[z_e] = \frac{2}{n} \cdot \frac{5}{12} + \frac{n-1}{4n} \cdot 1 + \left( 1 - \frac{2}{n} - \frac{n-1}{4n} \right) \cdot \frac{1}{3} = \frac{1}{2}
    \end{align*}

    Now, we prove $z'$ is in the spanning tree polytope. First, for any propert set $S \subseteq G$, we show $z'(E(S)) \leq |S| - 1$. Since for all $e$, $z'_e \leq z_e$, we show $z$ satisfies these constraints.
    
    By the definition of $z$, $z(\delta(r)) = \frac{5}{3}$, $z(\delta(v)) = \frac{25}{12}$ for a terminal vertex $v$, and, $z(\delta(v)) = 2$ for a non-terminal $v \in G \setminus r$. Now, consiser a proper subset $S \subset G$. Assume $S$ contains $i$ terminal vertices. We compute $z(E(S)) = \frac{  \sum_{v \in S} z(\delta(v)) - z(\delta(S))}{2}$.

    \begin{itemize}
        \item \textbf{Case 1}: $r \notin S$. Since $G$ is a degree-cut instance, $|\delta(S)| \geq 6$, $i$ of which have $z_e = \frac{5}{12}$. Therefore, $z(\delta(S)) \geq (6-i) \cdot \frac{1}{3} + i \cdot \frac{5}{12} = 2 + i \cdot \frac{1}{12}$. Moreover, $\sum_{v \in S} z(\delta(v)) = (|S|-i) \cdot 2 + i \cdot \frac{25}{12}$ which concludes $z(E(S)) \leq |S| - 1$.

        \item  \textbf{Case 2}: $r \in S$. Similarly, $|\delta(S)| \geq 6$, therefore, $z(\delta(S)) \geq 2$. Meanwhile, $\sum_{v \in S} z(\delta(v)) \leq \frac{5}{3} + 4 \cdot \frac{25}{12} + (|S|- 4) \cdot 2 = 2 \cdot |S|$. Concluding $z(E(S)) \leq |S| -1$
    \end{itemize}

    Finally, it suffices to show $z(E) = n$ as $z'(E) = z(E)-1$. Recall that for $r$, $z(\delta(r)) = \frac{5}{3}$, for each terminal vertex $v$, $z(\delta(v)) = \frac{25}{12}$, and, for every non-terminal vertex $u$, $z(\delta(u)) = 2$.

    \begin{align*}
        z(E) = \frac{1}{2} \cdot \left( \frac{5}{3} + 4 \cdot \frac{25}{12} + (n-5) \cdot 2 \right) = n
    \end{align*}
\end{proof}

Now, we will sample a tree $T$ from the max entropy distribution over spanning trees with marginals equal to $z'_e$. Finally, we add $e^+$ to $T$ so that $\P{e \in T} = z_e$ for all $e \in E$. 

\begin{definition}
    Let $M$ be the maximum matching sampled according to $\nu$. Let $r$ be the single vertex not saturated by $M$, we refer to this vertex as the root. Moreover, we refer to the vertices adjacent to $r$ in $G$ as terminals.
\end{definition}

\begin{definition}
    An edge $e = uv \in G$ is called normal if $e \in M$ and $u, v$ are not terminal vertices.
\end{definition}

\begin{lemma}
    For every edge $e \in E$, $\P{e \text{ even at last} \mid e \text{ normal}} \ge \frac{16}{81}$.
\end{lemma}
\begin{proof}
    Consider the normal edge $e = uv$ and let $A = \delta(u) \setminus e$, $B = \delta(v) \setminus e$. By the the definition of normal edges, any $f \in A \cup B$ has $z_e = \frac{1}{3}$. 
    Let $A = \{ a, b, c \}$ and $B = \{f, g, h \}$, note that $A$ and $B$ are disjoint. Project $\mu$ to $\{ a, b, c, f, g, h \}$ to get a distribution with generating polynomial $p(x_a, x_b, x_c, x_f, x_g, x_h)$. Symmetrize $p$ to get $p'(t_a, t_b) = p(t_a, t_a, t_a, t_b, t_b, t_b)$. 
    Since $|A| = |B| = 3$, the degree of $t_a$ and $t_b$ is at most $3$ in $p'$. Applying \cref{cor:capacity}, where $d_1$ is the degree of $t_a$ and $d_2$ the degree of $t_b$, yields,
    \begin{align*}
        \mathbb{P}[A_T=B_T=1] = p'_{(1,1)} \ge \frac{d_1 (d_1-1)^{d_1-1}}{d_1^{d_1}} \cdot \frac{d_2 (d_2-1)^{d_2-1}}{d_2^{d_2}} \ge \frac{16}{81}
    \end{align*}
\end{proof}

 \paragraph{Construction of the $O$-join Solution.} Now, we describe the construction of the $O$-join solution.

First, let $y_e = \frac{z_e}{2}$. Now, for any normal edge $e$ that is even at last, update $z_e$ to $\frac{1}{6}$. Finally, let $y_e = \frac{1}{4}$ for $e \in \delta(r)$.

\begin{claim}
    $y_e$ is in the $O$-join polytope for $T$.
\end{claim}

\begin{proof}
    Consider a proper $S \subset V$. Since $G$ is a degree cut instance, $|\delta(S)| \geq 6$. In our construction, for every edge $e \in E$, $y_e \geq \frac{1}{6}$. Therefore, $y(\delta(S)) \geq 1$. $y(\delta(r)) = 1$ by definition. Now, consider a single vertex $u \in G \setminus r$. When $u$ is odd, $y(\delta(u)) = \frac{z(\delta(u))}{2} \geq 1$. Thus every $O$-join constraint is satisfied.
\end{proof}

Finally, we analyze the expected cost of $y$.

\begin{lemma}
    \label{lem:degree-y-expected}
    For any vertex $u \in V$, $\mathbb{E}[y(\delta(u))] \leq \frac{227}{243} + O(\frac{1}{n})$.
\end{lemma}
\begin{proof}
    Consider the vertex $u$. Since $G$ is a $4$-regular graph, with probability at least $1 - \frac{17}{n}$, the shortest path from $u$ to $r$ contains at least three edges. Then, the matching edge $e$ that saturates $u$ is a normal edge and thus $y_e = \frac{1}{6}$ with probability at least $\frac{16}{81}$. Moreover, every other edge $f \in \delta(u)$ also has $y_f = \frac{1}{6}$ and $y(\delta(u)) = \frac{2}{3} \cdot \frac{16}{81} + \frac{65}{81} = \frac{227}{243} $. Similarly, with probability $\frac{1}{n}$, $r = u$ and $y(\delta(u)) = 1$. Otherwise, with probability $\frac{4}{n}$, $u$ is a terminal vertex and $y(\delta(u)) = \frac{1}{2} + \frac{1}{4} + 2 \cdot \frac{1}{6} = \frac{13}{12}$. Finally, if neither of these cases occur, $y(\delta(u)) \leq \frac{1}{2} + 3 \cdot \frac{1}{6} = 1$. Thus,
    \begin{align*}
        \mathbb{E}[y(\delta(u))] \leq \left(\frac{n-17}{n}\right)\cdot \frac{227}{243} + \frac{1}{n} + \frac{4}{n} \cdot \frac{13}{12} + \frac{12}{n} = \frac{227}{243} + O\left(\frac{1}{n}\right)
    \end{align*}
\end{proof}
 
The above lemma can also be proved by analyzing the value of $y_e$ for each edge. We analyze $y(\delta(u))$ instead simply to match the rest of this work. We are now ready to prove \cref{thm:degree-main}.

\begin{proof}[Proof of \cref{thm:degree-main}]
    By \cref{lem:degree-y-expected} and complementary slackness, the optimal $O$-join solution for $T$ has expected cost at most $\left( 0.4671 + O(\frac{1}{n}) \right) \cdot c(x)$. Moreover, by \cref{lem:degree-z-expected}, the expected cost of $T$ is at most $c(x)$, giving us a Eulerian circuit of expected cost at most $\left( 1.4671 + O(\frac{1}{n}) \right) \cdot c(x)$.
\end{proof}
We did not optimize the constant in the $O(\frac{1}{n})$ term. Second, we note that when $n$ is even, the term $O(\frac{1}{n})$ is not present.

\printbibliography

\end{document}